%% file: cdc_v7_ArXiv_.tex
\documentclass[journal,english]{IEEEtran}

\usepackage{epsfig}
\usepackage{times}
\usepackage{float}
\usepackage{afterpage}
\usepackage{sansmath}
\usepackage{amsmath}
\usepackage{amstext,cite}
\usepackage{amssymb,bm}
\usepackage{latexsym}
\usepackage{color}
\usepackage{graphicx}
\usepackage{amsmath}
\usepackage{amsthm}
\usepackage{graphicx}
\usepackage[center]{caption}
\usepackage{pstricks}
\usepackage{caption}
\usepackage{subcaption}
\usepackage{booktabs}
\usepackage{multicol}
\usepackage{lipsum}
\usepackage[T1]{fontenc}
\usepackage{hyperref}
\usepackage{aecompl}
\usepackage{mathrsfs}
\usepackage{blkarray}
\usepackage{cite}
\usepackage{arydshln}
\usepackage{algorithm}
\usepackage{algorithmic}
\usepackage{float} 
\usepackage{bbm}
\usepackage{empheq}

\usepackage{soul}
\usepackage{cancel}

\usepackage{changes}

\allowdisplaybreaks

\input{macros}

\usepackage{tikz}
\usetikzlibrary{calc}

\usepackage[nodisplayskipstretch]{setspace} 

\newcommand{\tikzmark}[1]{\tikz[overlay,remember picture] \node (#1) {};}
\newcommand{\DrawBox}[4][]{%
    \tikz[overlay,remember picture]{%
        \coordinate (TopLeft)     at ($(#2)+(-0.2em,0.9em)$);
        \coordinate (BottomRight) at ($(#3)+(0.2em,-0.3em)$);
        \path (TopLeft); \pgfgetlastxy{\XCoord}{\IgnoreCoord};
        \path (BottomRight); \pgfgetlastxy{\IgnoreCoord}{\YCoord};
        \coordinate (LabelPoint) at ($(\XCoord,\YCoord)!0.5!(BottomRight)$);
        \draw [red,#1] (TopLeft) rectangle (BottomRight);
        \node [below, #1, fill=none, fill opacity=1] at (LabelPoint) {#4};
    }
}
\begin{document}
\begin{sloppypar}

\title{Fundamental Limits of Distributed Linearly Separable Computation under Cyclic Assignment}  
\author{
Wenbo~Huang, 
Kai~Wan,~\IEEEmembership{Member,~IEEE,} 
Hua~Sun,~\IEEEmembership{Member,~IEEE,} 
Mingyue~Ji,~\IEEEmembership{Member,~IEEE,}
Robert~Caiming~Qiu,~\IEEEmembership{Fellow,~IEEE}
and~Giuseppe~Caire,~\IEEEmembership{Fellow,~IEEE}
\thanks{
A short version of this paper   was presented at the 2023 IEEE International Symposium on Information Theory (ISIT)~\cite{huang2023ISITversion}.
}
\thanks{
W.~Huang,~K.~Wan, and R.~C.~Qiu are with the School of Electronic Information and Communications, Huazhong University of Science and Technology, 430074 Wuhan, China (e-mail: \{eric\_huang, kai\_wan, caiming\}@hust.edu.cn). The work of W.~Huang, K.~Wan, and R.~C.~Qiu was partially funded by the   National Natural
Science Foundation of China (NSFC-12141107).}
\thanks{
H.~Sun is with the  Department of Electrical Engineering, University of North Texas, Denton, TX 76203, USA (e-mail: hua.sun@unt.edu). The work of H. Sun was supported in part by NSF Awards 2007108 and 2045656.
}
\thanks{
M.~Ji is with the Electrical and Computer Engineering Department, University of Florida, Gainesville, FL 32611, USA (e-mail: mingyueji@ufl.edu). The work of M.~Ji was supported in part by CAREER Award 2145835 and NSF Award 2312227.}
\thanks{
 G.~Caire is with the Electrical Engineering and Computer
Science Department, Technische Universit\"at Berlin, 10587 Berlin, Germany
(e-mail:  caire@tu-berlin.de). The work of   G.~Caire was partially funded by the European Research Council under the ERC Advanced Grant N. 789190, CARENET.}
}

\maketitle

\begin{abstract}

This paper studies the master-worker distributed linearly separable computation problem, where the considered computation task, referred to as linearly separable function,  is a generic linear map. This model includes cooperative distributed gradient coding, real-time rendering, linear transforms, etc. as special cases. 
The computation task on $\Ksf$ datasets can be expressed as $\Ksf_{\rm c}$ linear combinations of $\Ksf$ messages, where each message is the output of an individual function on one dataset. 
In this distributed computing model, the $\Ksf$ datasets  are assigned to $\Nsf$ workers for computation. Due to the possible presence of stragglers, it is required that the master can obtain the desired computation task from the answers of any $\Nsf_{\rm r}$ out of $\Nsf$ workers. 
The computation cost is defined as the number of datasets assigned to each worker, while the communication cost is defined as the number of codewords that should be received. The objective is to characterize the optimal tradeoff between the computation and communication costs. 
A common way to assign the datasets to the workers is ``cyclic assignment’’. This has been considered in several theoretical works, as well as gradient coding etc. Motivated by its theoretical and practical relevance, in this paper we focus on the cyclic assignment and  solve the problem by determining the optimal computation/communication cost tradeoff  when $\Nsf=\Ksf$ and order optimal within a factor of $2$ otherwise. In particular, this paper proposes a new computing scheme with the cyclic assignment based on the concept of interference alignment, by treating each message which cannot be computed by a worker as an interference from this worker. 
Beyond the appealing order-optimality result, we also show that the proposed scheme achieves significant gains over the current state of the art in practice. 
Experimental results over Tencent Cloud show   the reduction of whole distributed computing process time of our scheme is up to  $72.8\%$ compared to the benchmark scheme which treats the computation on   $\Ksf_{\rm c}$ linear combinations as $\Ksf_{\rm c}$ individual computations.

\end{abstract}


\section{Introduction}
\label{sec:intro}

The emergence of cloud and edge computing with the applications of artificial intelligence and machine learning requires efficient resource utilization and poses challenges in performing large-scale computations on big data~\cite{6G:_The_Next_Horizon}. To accelerate the computing process, distributed computing technology is widely adopted by clients on cloud platforms~\cite{Secure_Distributed_Matrix_Computation_With_Discrete_Fourier_Transform}. 
Leading cloud computing platforms such as  Amazon Web Services (AWS)~\cite{cloud2011amazon}, Microsoft Azure~\cite{microsoft_azure}, and Google Cloud Platform~\cite{Google_cloud} are commonly used in real systems. Apache Spark~\cite{zaharia2010spark} and MapReduce~\cite{dean2008mapreduce} are some of the distributed computing frameworks that have attracted cutting-edge research~\cite{wang2022big_data_analytics,shi2020communicationAI,wangni2018gradient_sparsification}.

The performance of distributed systems can be adversely affected by the presence of slow or dropped computing nodes,  referred to as  {\it stragglers}, and limited communication bandwidth. 
The overall distributed computation is bottlenecked by stragglers when the master needs to wait for all nodes' computing results~\cite{Straggler_1, Straggler_2}. Moreover,  the large-scale data exchanged in distributed computing systems (among the computing nodes or between each computing node and the server) results in high communication cost
~\cite{CCG,Communication_cost_1,Communication_cost_2}.   
To jointly address these two challenges in distributed computing, coding was originally introduced in~\cite{Speeding_Up_Distributed} to efficiently accelerate distributed matrix multiplication against stragglers, where the authors smartly treat stragglers as erasures in communications, and then erasure codes such as Minimum Distance Separable (MDS) codes were proposed to be used in distributed computing. Following the seminal work~\cite{Speeding_Up_Distributed}, coding techniques have been widely used to improve the efficiency of various distributed computing tasks, such as gradient descent~\cite{pmlr-v70-tandon17a, CCG,2021adaptive}, linear transform~\cite{Short-dot,m=1,Same_problem}, and matrix multiplication~\cite{Straggler_2,yu2020straggler,yu2019lagrange}. These computing tasks find wide application in numerous areas of engineering, physics, and mathematics.

 In this paper, we consider the {\it linearly separable computation}~\cite{Separable_Functions}, which covers gradient descent and linear transform as special cases and can be seen as a type of matrix multiplication. It has wide  applications in  practical machine learning and engineering problems, such as  gradient descent, fourier transform, real-time rendering, hybrid expert system and mixture-of-experts~\cite{brigham1988fast, nussbaumer1981fast, akenine2019real, yu2021plenoctrees, manoharan2023study, gale2023megablocks} . The distributed linearly separable computation problem is formulated over a master-worker system, where the master wants to compute a function $f(D_1,\ldots,D_{\Ksf})$ of datasets $\{D_1,\ldots,D_{\Ksf}\}$  with the help of $\Nsf$ workers as shown in Fig.~\ref{fig:system_model_in_original}.  
The objective function of linear separable computation $f(D_1,\ldots,D_{\Ksf})$  can be  further written as:
\begin{align}
     f(D_1,\ldots,D_{\Ksf})=g(f_1(D_1),\ldots,f_{\Ksf}(D_{\Ksf}))=g(W_1,\ldots,W_{\Ksf}),
     \label{eq:linearly sep function}
\end{align} 
where $W_{k}=f_{k}(D_{k}), k \in [\Ksf]$, is the computation result of the function $f_{k}(\cdot)$ with input $D_k$.  In this function the computation operates on different datasets separately. In addition,   linearly separable computation considers that   $g(\cdot)$ is a linear mapping. Then $g(W_1,\ldots,W_{\Ksf})$ is formed by $\Ksf_{\rm c}$ linear combinations of the  messages $W_1,\ldots,W_{\Ksf}$, assuming that each message  contains $\Lsf$ symbols uniformly i.i.d. on some finite field  $\mathbb{F}_{\qsf}$. In other words, $g(W_1,\ldots,W_{\Ksf})$ can be written as a matrix multiplication, 
${\bf F} [W_1; \ldots; W_{\Ksf}]$, where $\Ksf_{\rm c}$ linear functions are required and the  dimension of ${\bf F}$ is $\Ksf_{\rm c} \times \Ksf$. 

\begin{figure}
    \centering
        \includegraphics[scale=0.7]{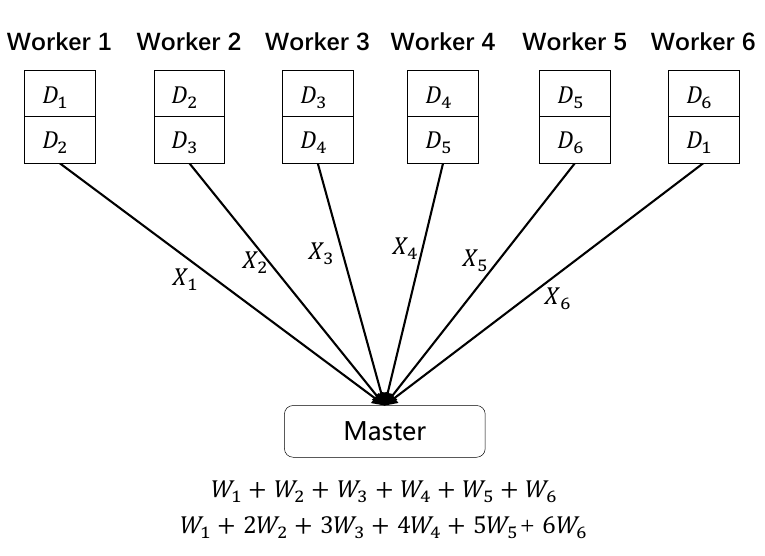}   
    \caption{\small  A distributed linearly separable computation system with$(\Ksf,\Nsf,\Nsf_{\rm r},\Ksf_{\rm c},\msf)=(6,6,6,2,2)$. }
      \label{fig:system_model_in_original}
\end{figure}

 The distributed linearly separable computing framework operates in three stages. First, the master assigns to each worker a subset of $\Msf = \frac{\Ksf}{\Nsf}(\Nsf - \Nsf_{\rm r} + \msf)$ datasets, ensuring robustness against up to $\Nsf - \Nsf_{\rm r}$ stragglers; a widely adopted approach is the cyclic assignment. Each worker then computes local messages from its assigned datasets and transmits a coded response to the master. Finally, the master decodes the desired linear combinations after receiving responses from any $\Nsf_{\rm r}$ workers. In conclusion, the framework focuses on characterizing the fundamental tradeoff between computation cost $\Msf$ and communication cost $\Rsf$. (The detailed system parameters will be discussed in Section~\ref{sec:system})

In the literature, various works have been proposed for the above problem with different system parameters:
\begin{itemize}
\item $\Ksf_{\rm c}=1$.
When $\Ksf_{\rm c}=1$, the distributed linearly computation problem reduces to the distributed gradient descent problem proposed in~\cite{pmlr-v70-tandon17a}. When the computation cost is minimum ($\Msf = \frac{\Ksf}{\Nsf}(\Nsf-\Nsf_{\rm r}+1)$), the optimal communication was characterized in~\cite{pmlr-v70-tandon17a}. For the general computation cost $(\Msf = \frac{\Ksf}{\Nsf}(\Nsf-\Nsf_{\rm r}+\msf), \msf > 1)$, distributed gradient descent schemes were proposed in~\cite{CCG, Short-dot}, which are optimal under the constraint of linear coding. 

\item $\Ksf_{\rm c}>1$. 
A distributed computing scheme was proposed in~\cite{m=1}, which is optimal under the constraint of the cyclic assignment, and is generally optimal under some system parameters with minimum computation cost.     
Under the constraint of  $\Nsf \geq \frac{{\msf + \usf - 1}}{\usf} + \usf\left( {{\Nsf_{\rm r}} - \msf - \usf + 1} \right)$ where $\usf = \left\lceil {\frac{{{{\Ksf_{\rm c}}\Nsf}}}{\Ksf}} \right\rceil$, a distributed computing scheme was proposed in~\cite{Same_problem} which is (order) optimal under the constraint of cyclic assignment with general computation cost. 
\end{itemize}  
Apart from the above mentioned special cases, the (order) optimality on the tradeoff between the computation and communication costs remains open, even under the constraint of cyclic assignment. 

Another critical remark is that the considered distributed linearly separable computation problem can be seen as a distributed one-side matrix multiplication problem with uncoded data assignment, where `one-side' means that the Right-Hand-Side (RHS) matrix in the multiplication is a matrix composed of datasets and the LHS matrix is the coefficient matrix known by all workers. Coded distributed matrix 
multiplication was widely studied in the literature (e.g. see~\cite{Straggler_2, yu2020straggler, yu2020coded_resilient, bitar2022adaptive}). However, most works (except the ones mentioned above) are built on coded data assignment, where linear combinations of $W_{1},..., W_{\Ksf}$ are assigned to each worker. 
These schemes could not be used in the considered problem since $W_{1},..., W_{\Ksf}$ are not raw datasets and cannot be assigned to the workers. Instead, we assign some of $D_{1},..., D_{\Ksf}$ to each worker, whose computation is using the assigned datasets; thus, the data assignment is uncoded.

\subsection{Main Contributions}
 For the distributed linearly separable computation problem under the cyclic assignment~\cite{Same_problem} with $\msf>1$ and $\Ksf_{\rm c}>1$,   we propose a novel computing scheme. Compared to the converse bound under the cyclic assignment in~\cite{Same_problem}, the proposed computing scheme is optimal if $\Ksf=\Nsf$, and order optimal within a factor of $2$ otherwise. Our scheme is based on the interference alignment strategy originally proposed in~\cite{2008Interference_Alignment}, by treating each message that one worker cannot compute as an `interference' to this worker. 
 Existing works on  distributed computing   schemes based on interference alignment~\cite{jia2019securedistributed,kakar2019distributedsecure,jia2019matrixCSA},  align the interferences at the receiver side (master node) (the undesired messages are treated as interferences to the receiver) such that the dimension of the linear space spanned by the received messages can be reduced. In contrast, the main novelty of our computing scheme is to use interference alignment at the transmitter side (worker nodes), such that in the transmission of each  worker the interferences are set to be $0$ while the number of transmissions is reduced.  
Our new achievability strategy essentially closes the optimal tradeoff problem for all system parameters. In addition, to demonstrate its effectiveness in practical scenarios, we have implemented the proposed scheme in Tencent Cloud and compared it with the benchmark that uses the computing scheme of~\cite{CCG} repeated $\Ksf_{\rm c}$ times (i.e., for each linear combination). Beyond the appealing order-optimality result, we show that the proposed scheme achieves significant gains over the current state of the art in practice. 
Experimental results show that the reduction of whole distributed computing process time of our scheme is up to  $72.8\%$ compared to the benchmark scheme.  
 


\subsection{Paper Organization}
The rest of this paper is organized as follows. Section~\ref{sec:system} presents the distributed linearly separable computation problem. Section~\ref{sec:main} provides the main result in this paper and some numerical evaluations. Section~\ref{sec:Achievable coding scheme} describes the proposed coding scheme. 
 Section~\ref{sec:Conclusion} concludes the paper, and some proofs are given in the Appendices.

\subsection{Notation Convention}
\label{sub:notation}
Calligraphic symbols denote sets, 
bold symbols denote vectors and matrices,
and sans-serif symbols denote system parameters.
We use $|\cdot|$ to represent the cardinality of a set or the length of a vector;
$[a:b]:=\left\{ a,a+1,\ldots,b\right\}$ and $[n] := [1:n]$; 
the sum of a set $\Sc$ and a scalar $ a$, $\Sc+a$ represents the resulting set by incriminating each element of $\Sc$ by $a$;
$\mathbb{F}_{\qsf}$ represents a  finite field with order $\qsf$;   
$\mathbf{A}^{\text{\rm T}}$  and $\mathbf{A}^{-1}$ represent the transpose  and the inverse of matrix $\mathbf{A}$, respectively;
$(\mathbf{A})_{m \times n}$ explicitly indicates that the matrix $\mathbf{A}$ is of dimension $m \times n$;
$\mathbf{A}^{(\Sc)_{\rm r}}$ represents the sub-matrix of $\mathbf{A}$ which is composed of the rows  of $\mathbf{A}$ with indices in $\Sc$ (here $\rm r$ represents `rows'); 
$\mathbf{A}^{(\Sc)_{\rm c}}$ represents the sub-matrix of $\mathbf{A}$ which is composed of the columns  of $\mathbf{A}$ with indices in $\Sc$ (here $\rm c$ represents `columns'); 
$\text{Mod}(b, a)$ represents the modulo operation on $b$ with integer quotient $a$ and in this paper we let $\text{Mod}(b, a)\in \{1, \ldots, a\}$;
let $a|b$ represent $b$ is divisible by $a$;
  In the rest of the paper, entropies will be in base $\qsf$, where $\qsf$ represents the field size.

\section{System Model}
\label{sec:system}
We consider the $(\Ksf,\Nsf,\Nsf_{\rm r},\Ksf_{\rm c})$ distributed linearly separable computation problem  over the canonical master-worker
distributed system, originally proposed in~\cite{m=1}. 
A master wants to compute a  function of $\Ksf$  datasets $D_1, \ldots, D_{\Ksf}$, with the help of $\Nsf$ workers.

With the assumption that the function is linearly separable from the datasets, the computation task can be   written as $\Ksf_{\rm c}\leq \Ksf$ linear combinations of $\Ksf$ messages
\begin{align}
    &f({D_1},{D_2}, \ldots ,{D_{\Ksf}}) = g({f_1}({D_1}), \ldots ,{f_{\Ksf}}({D_{\Ksf}}))  \nonumber\\
    &= g({W_1}, \ldots ,{W_{\Ksf}}) = {\bf F}[{W_1}; \ldots ;{W_\Ksf}] = [{F_1};\ldots;{F_{\Ksf_{\rm c}}}], \label{eq:computation task}
\end{align}
where the $i^{\text{th}}$ message is  ${W_i} = {f_i}({D_i}) $, representing  the outcome of the  component function $f_i(\cdot)$ applied to dataset $D_i$. We assume that  each message $W_i$ contains $\Lsf$ uniformly i.i.d. symbols in $\mathbb{F}_{\qsf}$, where $\qsf$ is large enough.\footnote{\label{foot:L large}  In this paper, we 
  consider the case  $\Ksf$ divides $\Nsf$. Note that the proposed schemes could be directly extended to the general case of $\Ksf, \Nsf$ by introducing some virtual datasets as proposed in~\cite{m=1}. And we assume that $\Lsf$ is large enough such
that any sub-message division is possible.} ${\bf F}$ represents the demand matrix with dimension $\Ksf_{\rm c} \times \Ksf$, where each element is uniformly i.i.d. over $\mathbb{F}_{\qsf}$, so that when $\Ksf_{\rm c} \le \Ksf,$  $\text{Rank}({\bf F}) = \Ksf_{\rm c}$ with high probability and when $\Ksf_{\rm c} > \Ksf,$ $\text{Rank}({\bf F}) = \Ksf$ with high probability, we only consider the non-trivial setting $\Ksf_{\rm c} \le \Ksf$. 

 The distributed computing framework comprises three phases, {\it data assignment, computing, and decoding}. 
 \paragraph*{Data assignment phase}
 We assign $\Msf:=\frac{\Ksf}{\Nsf} (\Nsf-\Nsf_{\rm r}+\msf)$ datasets to each worker. $\Zc_n$ denotes the set of indices of datasets assigned to worker $n$, where $\Zc_n \subseteq [\Ksf]$ and $|\Zc_n| = \Msf$.  We use $\Zc_n, n \in [\Nsf]$ to describe the data assignment.\footnote{\label{computation cost} In a distributed system, the complexity of computing the linear combinations of the messages is much lower than computing the separable functions. So the computation cost can be represented by the number of messages each worker computes.} 

\paragraph*{Computing phase}
Each worker $n \in [\Nsf]$ first computes the messages   $W_k = f_k(D_{k})$ where $k \in \Zc_n$. Then it computes the codeword $X_n$, which is a  function of the $\Msf$ messages $\{W_k:k\in \Zc_n\}$, where
$
      X_n = \varphi_{n} \big(\{W_k:k\in \Zc_n\}, {\bf F} \big)
$      
  with $    
      \varphi_{n} : \mathbb{F}_{\qsf}^{|\Zc_n|\Lsf} \times {\mathbb{F}_{\qsf}}^{{\Ksf_{\rm c}}\times \Ksf} \to {\mathbb{F}_{\qsf}}^{\Tsf_{n}},
$
and sends $X_n$ back to the master. The number of symbols in $X_n$ is denoted by $\Tsf_n$.

\paragraph*{Decoding phase}
The master only waits for the codewords from the first $\Nsf_{\rm r}$ arriving workers, and we call them active workers. Since the master cannot foresee the transmission environment and the computation capability of each worker, 
the computing scheme should be designed to tolerate any $\Nsf - \Nsf_{\rm r}$ stragglers. Thus, for each subset of workers $\Ac \subseteq [\Nsf]$ where $|\Ac| = \Nsf_{\rm r}$, by defining
$
    {X_{\Ac}}: = \{ {X_n}:n \in \Ac\}, 
$
there  should  be a decoding function such that 
$
\hat{g}_{\Ac}= \phi_{\Ac}\big( X_{\Ac}, {\bf F} \big),
$
where 
$
    {\phi _\Ac}: \mathbb{F}_{\qsf}^{\sum_{n \in \Ac}{\Tsf_{\rm n}}} \times {\mathbb{F}_{\qsf}}^{{\Ksf_{\rm c}}\times \Ksf} \to {\mathbb{F}_{\qsf}}^{{\Ksf_{\rm c}}\times \Lsf}.
$

The  worst-case     
error probability is denoted by
\begin{align}
 \varepsilon:= \max_{\Ac  \subseteq [\Nsf]: |\Ac|= \Nsf_{\rm r}} \Pr\{ \hat{g}_{\Ac} \neq g(W_1,   \ldots, W_{\Ksf}) \}. 
\end{align}
A computing scheme is called achievable if   the  worst-case     
error probability vanishes when $\qsf \to \infty$. 
 
We define 
\begin{equation}
    \Rsf: = \mathop {\max }\limits_{\Ac \subseteq [\Nsf]:|\Ac| = {\Nsf_{\rm r}}} \frac{{\sum\nolimits_{n \in \Ac} {{\Tsf_n}} }}{\Lsf}
\end{equation} 
as the communication cost, 
which represents the worst-case (normalized) number of symbols received by the master from any $\Nsf_{\rm r}$ active workers to recover the computation task.

\paragraph*{Objective}
The objective of the    $(\Ksf,\Nsf,\Nsf_{\rm r},\Ksf_{\rm c})$ distributed linearly separable computation problem is to characterize the optimal (minimum) communication cost  $\Rsf^{\star}$ among all achievable computing schemes. 
Note  that, in order to tolerate $\Nsf-\Nsf_{\rm r}$ stragglers, we should have $\msf \geq 1$~\cite{pmlr-v70-tandon17a}; i.e., $\msf\in [\Nsf_{\rm r}]$. 

\paragraph*{Cyclic assignment}
Under the cyclic assignment, each dataset $D_i$, where $i \in [\Ksf]$,  is assigned to workers with indices $\{\text{Mod}(i,\Nsf), \text{Mod}(i-1,\Nsf), \ldots, \text{Mod}(i - \Msf +1, \Nsf)\}$. Thus for each worker $n\in [\Nsf]$, we have  
\begin{align}
\Zc_n &=  \underset{p \in \left[0:  \frac{\Ksf}{\Nsf} -1 \right]}{\cup}   \big\{\text{Mod}(n,\Nsf)+ p  \Nsf , \text{Mod}(n+1,\Nsf)+ p  \Nsf , \ldots, \nonumber\\& \text{Mod}(n+\Nsf-\Nsf_{\rm r}+\msf,\Nsf)+ p  \Nsf  \big\}  \label{eq:cyclic assignment n divides k}
\end{align}
with cardinality of $\frac{\Ksf}{\Nsf} (\Nsf-\Nsf_{\rm r} +\msf)$.
 
For example, when $(\Ksf = 6,  \Nsf = 3, \Nsf_{\rm r} = 2, \msf = 1 ),$ we have $\Zc_{1} = \left\{1,2,4,5 \right\}, \Zc_{2} = \left\{2,3,5,6 \right\}, \Zc_{3} = \left\{1,3,4,6 \right\}.$
The optimal communication cost is denoted by $\Rsf^{\star}$.
The optimal communication cost under the cyclic assignment is denoted by $\Rsf^{\star}_{cyc}$.
In addition, the achievable communication cost by the proposed coding scheme in this paper is denoted by $\Rsf_{ach}.$ For ease of notation, we define $\usf:= \left\lceil \frac{\Ksf_{\rm c} \Nsf}{\Ksf}\right\rceil$ in the following.  
 
 A converse bound on $\Rsf^{\star}_{\text{cyc}}$ was proposed in~\cite{Same_problem}, which is reviewed as follows.
\begin{thm}[\cite{Same_problem}]
\label{thm:converse}
 For the  $(\Ksf,\Nsf,\Nsf_{\rm r}, \Ksf_{\rm c})$ distributed linearly separable computation problem under the cyclic assignment,
\begin{itemize}
\item when $\Ksf_{\rm c} \in \left[ \frac{\Ksf}{\Nsf} (\Nsf_{\rm r}-\msf+1) \right]$, we have 
\begin{subequations}
\begin{align}
\Rsf^{\star}_{\text{cyc}} \geq   \frac{ \Nsf_{\rm r} \Ksf_{\rm c}}{\msf+\usf-1}. \label{eq:case 1 converse}
\end{align} 
\item when $\Ksf_{\rm c} \in  \left( \frac{\Ksf}{\Nsf} (\Nsf_{\rm r}-\msf+1) :\Ksf \right] $, we have 
\begin{align}
\Rsf^{\star}_{\text{cyc}} \geq \Rsf^{\star}  \geq  \Ksf_{\rm c}. \label{eq:case 2 converse}
\end{align} 
\end{subequations}
\end{itemize}
\hfill $\square$
\end{thm}

The best achievable bound on the considered distributed linearly separable computation problem in the literature is the following. 
\begin{thm}[\cite{m=1,CCG,Same_problem}]
\label{thm:state of art achievable}
 For the  $(\Ksf,\Nsf,\Nsf_{\rm r}, \Ksf_{\rm c})$ distributed linearly separable computation problem, 
\begin{itemize}
\begin{subequations}
\item when $\msf =1$, the optimal communication cost can be achieved in~\cite{m=1}
\begin{align}
\Rsf_{\text{ach}} = \left\{ 
\begin{array}{rcl}
\Nsf_{\rm r}\Ksf_{\rm c}, &  & \text{when}~ \Ksf_{\rm c} \in \left[\frac{\Ksf}{\Nsf} \right];\\
 \frac{\Ksf}{\Nsf}\Nsf_{\rm r},   &   &\text{when}~   \Ksf_{\rm c} \in  \left( \frac{\Ksf}{\Nsf}: \frac{\Ksf}{\Nsf}\Nsf_{\rm r} \right]; \\
 \Ksf_{\rm c},  &    &\text{when}~    \Ksf_{\rm c} \in  \left(  \frac{\Ksf}{\Nsf}\Nsf_{\rm r} :  \Ksf \right];
\end{array}
\right. 
\end{align}

\item when $\Ksf_{\rm c}=1$, the optimal communication cost can be achieved in~\cite{CCG}
\begin{align}
\Rsf_{\text{ach}} = \frac{\Nsf_{\rm r}}{\msf}; 
\end{align}

\item when $\min\{\msf , \Ksf_{\rm c} \} >1$ and $40 \geq \Nsf \geq \frac{\msf+\usf-1}{\usf}+ \usf(\Nsf_{\rm r}-\msf-\usf+1)$, the following communication cost can be achieved in~\cite{Same_problem}
\begin{align}
\Rsf_{\text{ach}} = \left\{ 
\begin{array}{rcl}
\frac{\Nsf_{\rm r}\Ksf_{\rm c}}{\msf},    &\text{when}~  \Ksf_{\rm c} \in \left[ \frac{\Ksf}{\Nsf}\right]; \\
\frac{\Ksf \Nsf_{\rm r} \usf}{\Nsf ( \msf + \usf - 1)},    &\text{when}~   \Ksf_{\rm c} \in  \left( \frac{\Ksf}{\Nsf} : \frac{\Ksf}{\Nsf}(\Nsf_{\rm r} - \msf + 1) \right]; \\
\Ksf_{\rm c},  &\text{when}~    \Ksf_{\rm c} \in  \left( \frac{\Ksf}{\Nsf}(\Nsf_{\rm r} - \msf + 1) :  \Ksf \right].
\end{array}
\right. 
\end{align}
\end{subequations}
\end{itemize}
\hfill $\square$
\end{thm}

 Note that all schemes mentioned above in Theorem~\ref{thm:state of art achievable} are with the cyclic assignment. When $\min\{\msf, \Ksf_{\rm c}\} >1,$ the coded computing scheme in the literature can only work for a limited region ($\Nsf \geq \frac{\msf+\usf-1}{\usf}+ \usf(\Nsf_{\rm r}-\msf-\usf+1)$).  For the considered distributed linearly separable computation problem, not only the general (order) optimality remains quite open, but also the (order) optimality under the cyclic assignment has so far remained open.

\section{Main Results}
\label{sec:main}

\subsection{Main Results On Distributed Linearly Separable Computation}
\label{sub:separable}
Under the cyclic assignment given in~\eqref{eq:cyclic assignment n divides k}, in Section~\ref{sec:Achievable coding scheme} the scheme is instrumental to prove, constructively, the achievability of computation load - communication rate tradeoff given by the following theorem: 
\begin{thm}
\label{thm:main achievable scheme}
For the  $(\Ksf,\Nsf,\Nsf_{\rm r}, \Ksf_{\rm c})$ distributed linearly separable computation problem, where $\Nsf \leq 60$,
the following communication cost is achievable,  
\begin{itemize}
\item when  $\Ksf_{\rm c} \in \left[ \frac{\Ksf}{\Nsf}\right]$,  
\begin{subequations}
\begin{align}
  \Rsf_{\text{ach}}= \frac{\Nsf_{\rm r}\Ksf_{\rm c} }{\msf}; \label{eq:achie case 1}
\end{align} 
\item when  $ \Ksf_{\rm c}  \in \left( \frac{\Ksf}{\Nsf}: \frac{\Ksf}{\Nsf}(\Nsf_{\rm r}-\msf+1) \right]  $, 
\begin{align}
\Rsf_{\text{ach}}=\frac{\Nsf_{\rm r} \Ksf \usf}{\Nsf(\msf+\usf-1)}  ; \label{eq:achie case 2}
\end{align}

\item when  $ \Ksf_{\rm c} \in \left( \frac{\Ksf}{\Nsf}(\Nsf_{\rm r}-\msf+1) : \Ksf \right]$,
\begin{align}
\Rsf_{\text{ach}}=  \Ksf_{\rm c}. \label{eq:achie case 3}
\end{align}
\label{eq:achieve three cases}
\end{subequations}
\end{itemize}
\hfill $\square$
\end{thm}

Notice that when  $\Ksf_{\rm c} \in \left[\frac{\Ksf}{\Nsf} \right]$, as explained in~\cite{Same_problem}, to achieve~\eqref{eq:achie case 1}, we directly repeat the optimal computing schemes for $\Ksf_{\rm c}=1$ in~\cite{CCG}  $\Ksf_{\rm c}$ times. When $\Ksf_{\rm c} \in \left( \frac{\Ksf}{\Nsf}: \Ksf \right]$,  the coding scheme is based on interference alignment, and the decodability of the proposed scheme is proven in the Appendix~\ref{sec: proof of linear independence for each worker} based on the Schwartz-Zippel lemma~\cite{Schwartz, Zippel, Demillo_Lipton}. For the non-zero polynomial condition of the Schwartz-Zippel lemma, for the following two cases (i) $ \Ksf_{\rm c} \in \left[ \frac{\Ksf}{\Nsf}(\Nsf_{\rm r}-\msf+1) : \Ksf \right]$ and (ii) $\Nsf=\Nsf_{\rm r}$ and $\msf+\Ksf_{\rm c}-1$ divides $\Nsf$, we provide concrete proofs to show the non-zero polynomial condition, and thus prove our scheme is decodable. For the remaining cases, we numerically verify all cases up to $\Nsf \leq 60$,  and thus conjecture in the rest of the paper, {\bf the proposed computing scheme is decodable for any system parameters} with the communication cost given in Theorem~\ref{thm:main achievable scheme}.

By comparing the converse bound in Theorem~\ref{thm:converse} (proposed in~\cite{Same_problem}) and the achievable scheme in Theorem~\ref{thm:main achievable scheme}, it can obtain the following (order) optimal results under the constraint of the cyclic assignment, whose proof  is given in Appendix~\ref{sec: proof of compared}.
\begin{cor}
\label{thm: Compared}
    For the  $(\Ksf,\Nsf,\Nsf_{\rm r}, \Ksf_{\rm c})$ distributed linearly separable computation problem, where $\Rsf^{\star}$ is denoted as the optimal communication cost, 
\begin{itemize}
 \item when  $\Ksf = \Nsf$, we have 
\begin{subequations}
\begin{align}
  \Rsf_{\text{ach}} = \Rsf^{\star}_{\text{cyc}} =  \left\{ {\begin{array}{*{20}{c}}
{\frac{{{\Nsf_{\rm r}}{\Ksf_{\rm c}}}}{{\msf + \Ksf_{\rm c} - 1}},  ~\text{if}~{\Ksf_{\rm c}} \in \left[ {{\Nsf_{\rm r}} - \msf + 1} \right]};\\
{{\Ksf_{\rm c}}, ~\text{if}~{\Ksf_{\rm c}} \in \left[ {{\Nsf_{\rm r}} - \msf + 1 : \Ksf} \right]};
\end{array}} \right.
\end{align} 
\item when  $\Ksf_{\rm c} \in \left[\frac{\Ksf}{\Nsf} \right]$, we have   
\begin{align}
 \Rsf_{\text{ach}} = \Rsf^{\star}_{\text{cyc}} = \frac{\Ksf_{\rm c} \Nsf_{\rm r}}{\msf}; 
\end{align}
\item when  $ \Ksf_{\rm c}  \in \left( \frac{\Ksf}{\Nsf}: \frac{\Ksf}{\Nsf}(\Nsf_{\rm r}-\msf+1) \right]  $, we have 
\begin{align}
\Rsf_{\text{ach}} \leq \frac{2\Ksf_{\rm c}}{\frac{\Ksf}{\Nsf}\usf}\Rsf_{\text{ach}} \leq  2\Rsf^{\star}_{\text{cyc}}; 
\end{align}
\item when  $ \Ksf_{\rm c} \in \left( \frac{\Ksf}{\Nsf}(\Nsf_{\rm r}-\msf+1) : \Ksf \right]$, we have 
\begin{align}
\Rsf_{\text{ach}} = \Rsf^{\star}_{\text{cyc}} = \Rsf^{\star} = \Ksf_{\rm c}.
\end{align}
\end{subequations}
\end{itemize}
\hfill $\square$
\end{cor}
To summarize Theorem~\ref{thm: Compared}, we can see that the proposed computing scheme in Theorem~\ref{thm:main achievable scheme} is order optimal within a factor of $2$ under the constraint of the cyclic assignment. 
In addition, it can be seen that  the cyclic assignment is generally optimal when $ \Ksf_{\rm c}  \in \left[ \frac{\Ksf}{\Nsf}: \frac{\Ksf}{\Nsf}(\Nsf_{\rm r}-\msf+1) \right]  $  or $ \Ksf_{\rm c} \in \left( \frac{\Ksf}{\Nsf}(\Nsf_{\rm r}-\msf+1) : \Ksf \right]$.

 Besides the cyclic assignment, the  repetition assignment\footnote{\label{foot:rep division}Under the repetition assignment, the number of workers $\Nsf$ should be a multiple of $\Nsf - \Nsf_{\rm r} + \msf$. We divide the $\Nsf$ workers into $\frac{\Nsf}{\Nsf - \Nsf_{\rm r} + \msf}$ groups, each consisting of $\Nsf - \Nsf_{\rm r} + \msf$ workers. 
Within each group,  the same $\frac{\Ksf}{\Nsf}(\Nsf - \Nsf_{\rm r} + \msf)$ datasets are assigned to each worker, where there does not exist any common  dataset assigned to more than one group. 
Thus the datasets assigned to each worker  in group $i$ is 
\begin{align*}
&\Zc'_i:= \underset{p \in \left[0:  \frac{\Ksf}{\Nsf} -1 \right]}{\bigcup}   \big\{\text{Mod}(1+(\Nsf - \Nsf_{\rm r} + \msf)(i-1),\Nsf)+ p  \Nsf ,\\& \text{Mod}(2+(\Nsf - \Nsf_{\rm r} + \msf)(i-1),\Nsf)+  p  \Nsf , \ldots,  \text{Mod}((\Nsf - \Nsf_{\rm r} + \msf)+ \\&   (\Nsf - \Nsf_{\rm r} + \msf)(i-1),\Nsf)+ p  \Nsf  \big\}. 
\end{align*}}  has also been used in  the  distributed gradient descent problem  and related problems~\cite{pmlr-v70-tandon17a,replicationcode2020}. In the following, we provide  the optimal communication cost under the repetition assignment.
\begin{thm} \label{cor: rep assignment}
For the  $(\Ksf,\Nsf,\Nsf_{\rm r}, \Ksf_{\rm c})$ distributed linearly separable computation problem with $ (\Nsf - \Nsf_{\rm r} + \msf) | \Nsf$,
 the optimal communication cost under the repetition assignment  is  
\begin{align}
    \Rsf_{\text{rep}}  = \frac{\Nsf_{\rm r}\min\{\Ksf_{\rm c}, \frac{\Ksf}{\Nsf}(\Nsf -\Nsf_{\rm r} +\msf)\}}{\msf}. \label{eq:convese rep 0}
\end{align}


We then compare the proposed computing scheme under the cyclic assignment with the optimal computing scheme under the repetition assignment: 
\begin{itemize}
\begin{subequations}
    \item when  $\Nsf_{\rm r} = \msf$, we have 
\begin{align}
  \Rsf_{\text{ach}} = \Rsf_{\text{rep}} =  \Ksf_{\rm c};
\end{align} 
    \item when  $\Ksf_{\rm c} \in \left[\frac{\Ksf}{\Nsf} \right]$, we have   
\begin{align}
 \Rsf_{\text{ach}} = \Rsf_{\text{rep}} = \frac{\Ksf_{\rm c} \Nsf_{\rm r}}{\msf}; 
\end{align}
    \item when  $ \Nsf_{\rm r} \neq  \msf, \Ksf_{\rm c} \in \left( \frac{\Ksf}{\Nsf}: \frac{\Ksf}{\Nsf}(\Nsf-\Nsf_{\rm r}+\msf) \right]  $, we have 
\begin{align}
    \frac{\Rsf_{\text{ach}}}{\Rsf_{\text{rep}}} = \frac{\msf}{\msf + \usf - 1} \le 1, \label{neq: 1}
\end{align}
since $\usf \ge 1$;
    \item when  $ \Nsf_{\rm r} \neq  \msf$, $\Ksf_{\rm c} \in \left(  \frac{\Ksf}{\Nsf}(\Nsf-\Nsf_{\rm r}+\msf): \right.$ $\left. \frac{\Ksf}{\Nsf}(\Nsf_{\rm r} - \msf + 1) \right]$\footnote{\label{foot: divide}Note that $\Nsf_{\rm r} \neq \msf,$ we have $\Nsf - \Nsf_{\rm r} + \msf \neq \Nsf,$ and $(\Nsf - \Nsf_{\rm r} + \msf) | \Nsf$, hence, $\Nsf - \Nsf_{\rm r} + \msf \le \frac{\Nsf}{2}.$ Thus, we have $\Nsf-\Nsf_{\rm r}+\msf \le \frac{\Nsf}{2} \le \Nsf - (\Nsf-\Nsf_{\rm r}+\msf) \textless \Nsf_{\rm r} - \msf + 1.$}, we have 
\begin{align}
 \frac{\Rsf_{\text{ach}}}{\Rsf_{\text{rep}}} = \frac{\usf\msf}{(\msf+ \usf- 1)(\Nsf -\Nsf_{\rm r} + \msf)} \le 1, \label{neq: 2}
\end{align}
since $\frac{\usf}{\msf+\usf-1} \le 1$ and $\frac{\msf}{\Nsf -\Nsf_{\rm r}+\msf} \le 1$;
    \item when  $ \Nsf_{\rm r} \neq  \msf, \Ksf_{\rm c} \in \left(  \frac{\Ksf}{\Nsf}(\Nsf_{\rm r} - \msf + 1) : \Ksf \right]$, we have 
\begin{align}
 \frac{\Rsf_{\text{ach}}}{\Rsf_{\text{rep}}} = \frac{\usf\msf}{\Nsf_{\rm r}(\Nsf - \Nsf_{\rm r} +\msf)} \leq 1,  \label{neq: 3}
\end{align}
since $
    \Nsf_{\rm r}(\Nsf-\Nsf_{\rm r}+\msf) - \usf\msf \ge \Nsf_{\rm r}(\Nsf-\Nsf_{\rm r}+\msf) - \Nsf\msf =(\Nsf_{\rm r} - \msf)(\Nsf -\Nsf_{\rm r}) 
    \ge 0. 
$
\end{subequations}
\end{itemize}
\hfill $\square$
\end{thm}

In one word,  the proposed computing   scheme under the cyclic assignment generally requires a lower communication cost than the scheme under repetition assignment, while the multiplicative gap between them can be unbounded.

\subsection{Numerical Evaluations}
We provide some numerical evaluations of the proposed computing scheme. 
The benchmark scheme is to  repeat the computing scheme in~\cite{CCG} $\Ksf_{\rm c}$ times, whose communication cost is thus  $ \frac{\Nsf_{\rm r} \Ksf_{\rm c}}{\msf}$.  It can be seen from Fig.~\ref{fig:m=3} and Fig.~\ref{fig:Kc=4} that our proposed scheme significantly reduces the communication cost compared to the benchmark scheme and the computing scheme under repetition assignment in Theorem~\ref{cor: rep assignment}, while coinciding with the converse bound under the cyclic assignment in Theorem~\ref{thm:converse} when $\Ksf=\Nsf$. 

\begin{figure}[t]
    \centering
    \begin{subfigure}[t]{\linewidth}
        \centering
        \includegraphics[scale=0.5]{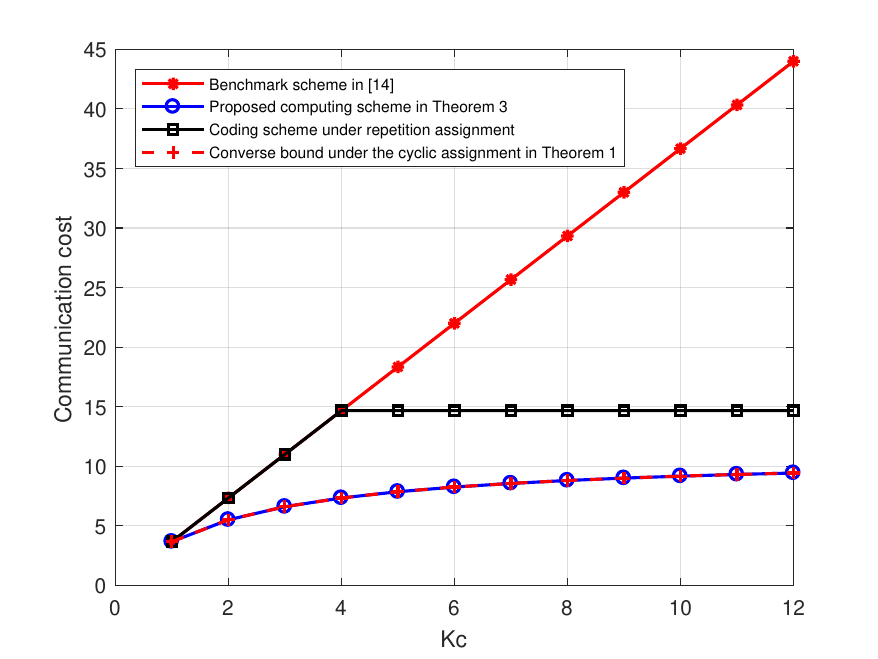}
        \caption{\small Communication costs for $(\Ksf=\Nsf=12, \Nsf_{\rm r} =11, \msf = 3, \Ksf_{\rm c}\in [12])$.}
        \label{fig:m=3}
    \end{subfigure}
    \vskip 0.5em
    \begin{subfigure}[t]{\linewidth}
        \centering
        \includegraphics[scale=0.5]{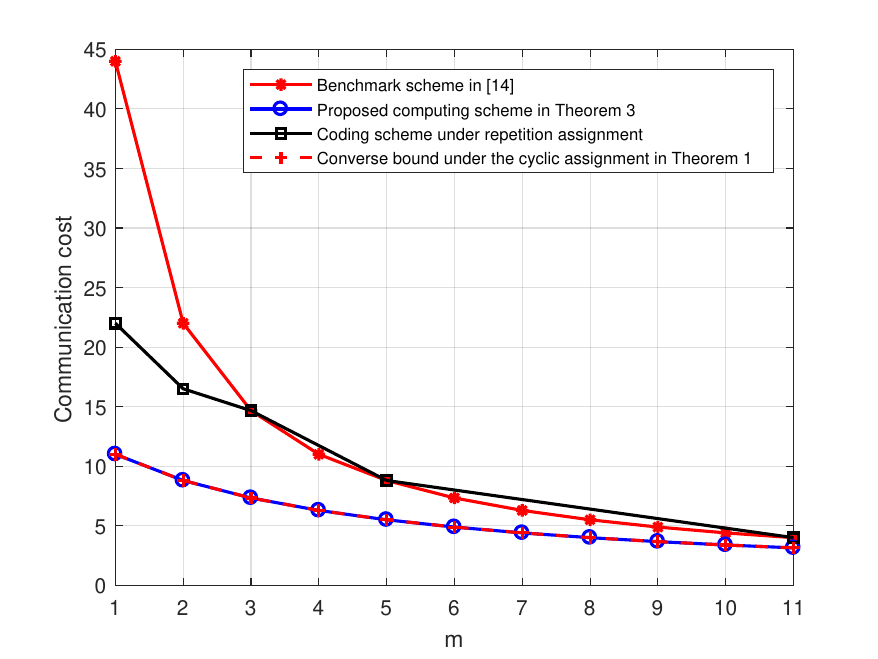}
        \caption{Communication costs for $(\Ksf=\Nsf=12, \Nsf_{\rm r} =11, \msf\in [11], \Ksf_{\rm c} =4)$.}
        \label{fig:Kc=4}
    \end{subfigure}
    \caption{\small Numerical evaluations for $\Ksf=\Nsf=12, \Nsf_{\rm r} =11$.}
    \label{fig:Simulation Result}
\end{figure}

\section{Proof of Theorem~\ref{thm:main achievable scheme}: New achievable distributed computing scheme}
\label{sec:Achievable coding scheme}
In this section, the proposed distributed coding scheme under the cyclic assignment in Theorem~\ref{thm:main achievable scheme} will be described based on the strategy of interference alignment. 

We will consider three cases, $\Ksf_{\rm c} \in \left[\frac{\Ksf}{\Nsf} \right]$, $\Ksf_{\rm c}\in \left( \frac{\Ksf}{\Nsf}(\Nsf_{\rm r}-\msf+1 ) : \Ksf \right]$, and $\Ksf_{\rm c}\in \left( \frac{\Ksf}{\Nsf} : \frac{\Ksf}{\Nsf}(\Nsf_{\rm r}-\msf+1) \right]$, respectively. Note that our major contribution is on the case 
  $\Ksf_{\rm c}\in \left( \frac{\Ksf}{\Nsf}  : \frac{\Ksf}{\Nsf}(\Nsf_{\rm r}-\msf+1) \right]$ to prove~\eqref{eq:achie case 2},
  since the computing scheme for $\Ksf_{\rm c} \in \left[\frac{\Ksf}{\Nsf} \right]$ is obtained by repeating the computing scheme in~\cite{CCG} $\Ksf_{\rm c}$ times, and the computing scheme for $\Ksf_{\rm c}\in \left( \frac{\Ksf}{\Nsf}(\Nsf_{\rm r}-\msf+1 ) : \Ksf \right]$ is obtained by leveraging the computing scheme for $\Ksf_{\rm c}=\frac{\Ksf}{\Nsf}(\Nsf_{\rm r}-\msf+1)$.

\subsection{$\Ksf_{\rm c} \in \left[\frac{\Ksf}{\Nsf} \right]$}

When  $\Ksf_{\rm c} \in \left[\frac{\Ksf}{\Nsf} \right]$, to achieve $\Rsf_{\text{ach}} = \frac{\Ksf_{\rm c} \Nsf_{\rm r}}{\msf}$, we directly repeat the optimal computing schemes for $\Ksf_{\rm c}=1$ in~\cite{CCG} (whose communication cost is $\frac{ \Nsf_{\rm r}}{\msf}$)  $\Ksf_{\rm c}$ times. 

\subsection{$\Ksf_{\rm c}\in \left( \frac{\Ksf}{\Nsf}(\Nsf_{\rm r}-\msf+1 ) : \Ksf \right]$}

\label{sub: extension part}
As   explained  in~\cite[Section IV-C]{m=1}, if the proposed scheme works for the case $\Ksf_{\rm c}=\frac{\Ksf}{\Nsf}(\Nsf_{\rm r}-\msf+1 )$ with communication cost $\Ksf_{\rm c}$, then it also works for the case 
$\Ksf_{\rm c} \in \left[\frac{\Ksf}{\Nsf}(\Nsf_{\rm r}-\msf+1 ): \Ksf \right]$ with communication cost $\Ksf_{\rm c}$ which coincides with~\eqref{eq:achie case 3}. 

More precisely, 
consider a system with parameters $(\Nsf , \Ksf , \Nsf_{\rm r }, \Ksf_{\rm c},  \msf )$, where $\Ksf_{\rm c} \in ( \frac{\Ksf}{\Nsf}(\Nsf_{\rm r}-\msf+1 ) :$ $ \Ksf ]$.  We divide $W_{i}$  where $i \in [\Ksf]$ into  $\binom{{\Ksf_{\rm c}} - 1}{\frac{\Ksf}{\Nsf}\left({{\Nsf_{\rm r}} - \msf + 1} \right)- 1 }$ non-overlapping and equal-length sub-messages and treat the problem into $\binom{{{\Ksf_{\rm c}} }}{\frac{\Ksf}{\Nsf}\left( {{\Nsf_{\rm r}} - \msf + 1} \right) }$ sub-problems with the parameters $(\Nsf , \Ksf , \Nsf_{\rm r },\Ksf_{\rm c} = \frac{\Ksf}{\Nsf}(\Nsf_{\rm r}-\msf+1 ),  \msf)$.  When each sub-problem is solved with the communication cost $\frac{\Ksf_{\rm c}}{\binom{{{\Ksf_{\rm c}} }}{\frac{\Ksf}{\Nsf}\left( {{\Nsf_{\rm r}} - \msf + 1} \right) }}$, the master can recover the task demand and the total communication cost is $\Ksf_{\rm c}$ as same as~\eqref{eq:achie case 3}. 

\subsection{$\Ksf_{\rm c}\in \left( \frac{\Ksf}{\Nsf} : \frac{\Ksf}{\Nsf}(\Nsf_{\rm r}-\msf+1) \right]$}

Let us then focus on the most important case and propose a new computing scheme based on interference alignment. We introduce the main ideas of the proposed scheme through the following example. 

\begin{example}[$(\Ksf, \Nsf, \Nsf_{\rm r}, \Ksf_{\rm c}, \msf)=(6,6,6,2,2)$]
\label{ex:6622}
\rm 
Consider the example where $\Nsf = \Ksf = 6 $, $\Nsf_{\rm r} = 6$, $\Ksf_{\rm c} = \usf = 2$, and $\Msf=\msf = 2$, as illustrated  in Fig.~\ref{fig:system_model_in_original}. In this example, since $\Nsf_{\rm r} =\Nsf$, there is no straggler.

The converse bound under the cyclic assignment in Theorem~\ref{thm:converse} for this example is $\Rsf^{\star}_{\text{cyc}}\geq \frac{ \Nsf_{\rm r} \Ksf_{\rm c}}{\msf+\usf-1}= 4$.  Without loss of generality, we assume that the demand matrix $\bf F$ is
\begin{align}
   {\bf F} = \left[ {\begin{array}{*{20}{c}}
1&1&1&1&1&1\\
1&2&3&4&5&6
\end{array}} \right].
\end{align} 
Note that the computing scheme in~\cite{Same_problem} cannot work for this example.

\paragraph*{Data assignment phase}
We assign the datasets with the cyclic assignment as shown in   Table~\ref{Table 3}. Here, $\Zc_{n} = \{n,\text{Mod}(n+1,6)\}$ for each $n \in [5]$.
\setlength{\tabcolsep}{1mm}{
\begin{table}[!htbp]
\caption{Data assignment}\label{Table 3}
\centering
\begin{tabular}{|c|c|c|c|c|c|}
\hline
{worker 1}&{worker 2}&{worker 3}&{worker 4}&{worker 5}&{worker 6}\\ 
\hline
{$D_1$}&{$D_2$}&{$D_3$}&{$D_4$}&{$D_5$}&{$D_6$}\\

{$D_2$}&{$D_3$}&{$D_4$}&{$D_5$}&{$D_6$}&{$D_1$}\\
\hline
\end{tabular}
\end{table}
}  
\paragraph*{Computing phase}
To attain the  converse bound $\Rsf^{\star}_{\text{cyc}} \geq \frac{ \Nsf_{\rm r} \Ksf_{\rm c}}{\msf+\usf-1}= 4$, we divide each message $W_k$ where $k \in [6]$ into $\dsf := \msf+\usf-1= 3$ non-overlapping and equal-length
sub-messages, $W_{k}=\{W_{k,j}:j\in [\dsf]\}$, where each sub-message contains $\frac{\Lsf}{\dsf}=\frac{\Lsf}{3}$ symbols in $\mathbb{F}_{\qsf}$.  
Denote ${\bf W} = \left[ W_{1,1}; \ldots;W_{6,1};W_{1,2};\ldots;W_{6,3} \right]$. When each worker transmits $\frac{\Ksf_{\rm c}\Lsf}{\dsf} = \frac{2\Lsf}{3}$ symbols (or $\Ksf_{\rm c}=2$ linear combinations of sub-messages), the communication cost will coincide with the converse bound.

After the message partitioning, the computation task becomes $\Ksf_{\rm c} \times \dsf = 6$ linear combinations of sub-messages. Since there are totally $\Nsf_{\rm r} \times \Ksf_{\rm c} = 12$ linear combinations transmitted from workers, the effective demand ( linear space received by the master) could be expressed as ${\bf F}'{\bf W} $, 
where the dimension of ${\bf F}'$ is $\Nsf_{\rm r}\Ksf_{\rm c}\times \Ksf\dsf = 12\times 18$, with the form  
\begin{equation}
 \mathbf{F}' =\left[\begin{array}{c:c:c}
 ({\bf F})_{2 \times 6}  & {\bf 0}_{2 \times 6}  & {\bf 0}_{2 \times 6}   \\ \hdashline
{\bf 0}_{2 \times 6} &  ({\bf F})_{2 \times 6}   & {\bf 0}_{2 \times 6}   \\  \hdashline
 {\bf 0}_{2 \times 6} &   {\bf 0}_{2 \times 6}   &  ({\bf F})_{2 \times 6} \\ \hdashline
 ({\bf V_ {1}})_{6 \times 6}  &  ({\bf V_ {2}})_{6 \times 6}   &   ({\bf V_ {3 }})_{6 \times 6} 
 \end{array}
\right],
\end{equation}
 where we denote $\left[{\bf V_ {1}}, {\bf V_ {2}}, {\bf V_ {3}} \right]$ as the virtual demands, which can be designed carefully to ensure the optimal transmission.    
  
Our  coding strategy is to let each worker $n\in [6]$ send $\Ksf_{\rm c} = 2$ linear combinations of sub-messages (codewords) 
\begin{align}
\sv^{n,1} {\bf F}' {\bf  W} \text{ and } 
\sv^{n,2} {\bf F}' {\bf W}, \label{eq:transmitted linear combinations}
\end{align}  
totally $\frac{\Ksf_{\rm c} \Lsf}{\dsf}=\frac{2\Lsf}{3}$ symbols,  
   coinciding with the converse bound. 
 Note that $\sv^{n,1} $ and $\sv^{n,2}$ are vectors with length $\Nsf_{\rm r}\Ksf_{\rm c} =12$. 
Denote the overall   coding matrix by   
\begin{align}
    \mathbf{S}= \left[\sv^{1,1};\sv^{1,2};\sv^{2,1};\ldots;\sv^{6,2}\right]. \label{eq:ex2 coding matrix S}
\end{align} 
Note that worker $n$ does not have $\{D_k : k \in \overline{\Zc_{n}}\}$ where $\overline{\Zc_{n}} = [6] \setminus \Zc_{n}$. Consequently, it cannnot compute the sub-messages
$\{W_{k,j}: k \in \overline{\Zc_{n}}, j \in [3]\}$, and thus the coefficients corresponding to these sub-messages should be zero in~\eqref{eq:transmitted linear combinations}. (i.e. the $i^{\text{th}}$ element of $\sv^{n,1} {\bf F}'$ and the $i^{\text{th}}$ element of $\sv^{n,2} {\bf F}'$ are $0$, where $i\in  \overline{\Zc_{n}} \cup (\overline{\Zc_{n}}+6) \cup (\overline{\Zc_{n}}+12)$. )
  
The column-wise sub-matrix of ${\bf F}'$ corresponding to the sub-messages that are unavailable to worker $n$ is given by 
\begin{equation}
  {\bf F}'(n) = \left[ {\begin{array}{*{20}{c}}
\overline{ {\bf F}_{n} }&{\bf 0} &{\bf 0}\\
{\bf 0} & \overline{ {\bf F}_{n} } & {\bf 0}\\
{\bf 0} & {\bf 0} &\overline{ {\bf F}_{n} }\\
\overline{{\bf V}_{1,n}}& \overline{{\bf V}_{2,n}} &\overline{{\bf V}_{3,n}}
\end{array}} \right]_{12 \times 12}, \label{eq: Fn2}   
\end{equation}
where $\overline{{\bf F}_{n}}$ and $\overline{{\bf V}_{i,n}}$ for $i \in [\dsf]$ and $n \in [\Nsf]$, denote the column sub-matrices of ${\bf F}$ and ${\bf V}_{i}$, respectively, containing the columns indexed by $\overline{\Zc_n}$. Then we obtain the following condition 
$\sv^{n,1} {\bf F}'(n) =\sv^{n,2}{\bf F}'(n)= {\bf 0}_{1\times 12}$.
 
Combining with the decodability condition, the objective of the coding scheme is to design the virtual demand ${\bf V}_ {1}, {\bf V}_ {2}$ and ${\bf V}_ {3}$, such that    
the following two conditions are satisfied: 
\begin{itemize}
\item (c1)  for each $n\in [6]$, the rank of ${\bf F}'(n)$ is $10$, such that ${\bf F}'(n)$ contains $ 2$ linearly independent left-zero vectors; let $\sv^{n,1} $ and $\sv^{n,2}$ be these two vectors;
\item (c2)   the matrix ${\bf S}$ defined in~\eqref{eq:ex2 coding matrix S} is of full rank, such that from the answers of all workers, the master can recover $ \mathbf{F}' {\bf W}$.
\end{itemize}

Next, we present the proposed coding scheme from the perspective of interference alignment. 
If we choose the elements in the virtual demands of ${\bf F}'(n)$ uniformly i.i.d. in $\mathbb{F}_{\qsf}$ as the computing scheme in~\cite{m=1},  ${\bf F}'(n)$ is of full rank with high probability, and there is no left null vector; thus 
worker $n$ cannot transmit two linearly independent codewords in terms of its known datasets. 
In order to let worker $n$ transmit two linearly independent codewords, we propose to use interference alignment strategy~\cite{2008Interference_Alignment}. 
More precisely, for each worker $n\in [6]$, we treat the sub-messages it cannot compute as its interferences, i.e., $W_{k,j}$ where $k\in \overline{\Zc_{n}}$ and $j\in [3]$. The column-wise sub-matrix of the effective demand matrix corresponding to these interferences is ${\bf F}'(n)$, with dimension $12 \times 12$. Now we want to `align' the interferences such that the rank of ${\bf F}'(n)$ is $10$ and ${\bf F}'(n)$ contains two linearly independent left null vectors. 

In this example, the interference alignment should be designed to cross different workers to satisfy Conditions (c1) and (c2). 
Our scheme is based on the following key observation: 
under the cyclic assignment, each neighbouring $\Nsf_{\rm r}-\msf-\Ksf_{\rm c} = 2$ workers do not have  $\Ksf_{\rm c}+1=3$ common datasets. 
For example, both workers  $1$ and $2$ do not have $D_{4}, D_5,D_6$.  
Hence,  if we generate one linear reduction equation $\ev_{1}$ on the columns of ${\bf F}'$ with indices in $\{4,5,6,10,11,12,16,17,18\}$, we can reduce the ranks of ${\bf F}'(1)$ and ${\bf F}'(2)$ by $1$ simultaneously. More precisely, let 
$ \ev_1 = (0, 0, 0, e_{1,4}, e_{1,5}, e_{1,6}, 0, 0, 0,e_{1,10}, e_{1,11}, e_{1,12}, 0, 0, 0, e_{1,16}, $ $ e_{1,17}, e_{1,18})$, satisfying 
\begin{align}
{\bf F}' \ev_1^{\text{\rm T}}= {\bf 0}_{12\times 1}. \label{eq:one linear eq}
\end{align}
 In addition, by the first two rows of ${\bf F'}$, it can be seen that $(e_{1,4},e_{1,5},e_{1,6})$ should be a multiple of $(1,-2,1)$, since $\left[ {\begin{array}{*{20}{c}}
 1&1&1\\
 4&5&6
\end{array}} \right] [1,-2,1]^{\text{\rm T}}={\bf 0}_{2\times 1} $.  Similarly,  by the third and fourth   rows of ${\bf F'}$, $(e_{1,10},e_{1,11},e_{1,12})$ 
should also be  a multiple of $(1,-2,1)$; by the fifth and sixth   rows of ${\bf F'}$, $(e_{1,16},e_{1,17},e_{1,18})$
should also be  a multiple of $(1,-2,1)$. 
  Hence, by randomly selecting the multiple numbers, we let $(e_{1,4},e_{1,5},e_{1,6})=(1,-2,1)$, $(e_{1,10},e_{1,11},e_{1,12})=(1, -2, 1)$, and $(e_{1,16},e_{1,17},e_{1,18})=(2, -4, 2)$.
Thus we have 
$$
\ev_{1}= (0, 0, 0, 1, -2, 1, 0, 0, 0, 1, -2, 1, 0, 0, 0, 2, -4, 2).
$$

 Similarly, for each $i\in [6]$, the workers in $\{i,\text{Mod}(i+1,6)\}$ do not have $D_{\text{Mod}(i-1,6)}, D_{\text{Mod}(i-2,6)}, $ $D_{\text{Mod}(i-3,6)}$; we generate a vector $\ev_{i}$, whose elements with indices in $\{\text{Mod}(i-1,6),\text{Mod}(i-2,6),\text{Mod}(i-3,6),\text{Mod}(i-1,6)+6,\text{Mod}(i-2,6)+6,\text{Mod}(i-3,6)+6,\text{Mod}(i-1,6)+12,\text{Mod}(i-2,6)+12,\text{Mod}(i-3,6)+12\}$ are not all $0$ and the other elements are all $0$. 
By ${\bf F}' \ev_i^T= {\bf 0}_{12\times 1}$ and the structure of ${\bf F}'$, the vectors 
\begin{subequations}
    \begin{align}
        &(e_{i,\text{Mod}(i-3,6)},e_{i,\text{Mod}(i-2,6)},e_{i,\text{Mod}(i-1,6)}),\\
        &(e_{i,\text{Mod}(i-3,6)+6},e_{i,\text{Mod}(i-2,6)+6},e_{i,\text{Mod}(i-1,6)+6}),\\
        &(e_{i,\text{Mod}(i-3,6)+12},e_{i,\text{Mod}(i-2,6)+12},e_{i,\text{Mod}(i-1,6)+12})
    \end{align}
    \label{eq:ex2 three null vectors}
\end{subequations}
 should be   left null vectors  of 
 \begin{align}
     \begin{bmatrix} 
1 & 1 & 1  \\
 \text{Mod}(i-3,6) & \text{Mod}(i-2,6) & \text{Mod}(i-1,6)    \\
\end{bmatrix}.\label{eq:matrix111}
 \end{align}
Note that the matrix in~\eqref{eq:matrix111} admits only one linearly independent left null vector. Consequently, each vector in~\eqref{eq:ex2 three null vectors} must be a scalar multiple of this unique null vector. The matrix ${\bf E} = [\ev_1;\ev_2;\ldots;\ev_6]$ is explicitly constructed in~\eqref{eq:E for example_2}. Each vector $\ev_i$, for $i \in [6]$, is designed such that it simultaneously reduces the interference dimension by one for a specific pair of workers. Moreover, each worker participates in exactly two such worker pairs. As a result, the overall interference reduction effect can be captured by the following relation:
\begin{align}
{\bf F}' {\bf E}^{\text{\rm T}} = {\bf 0}_{12\times 6}, \label{eq:equations for example_2}
\end{align}
which implies that the interference dimension is reduced by two for each worker.

 The next step is to solve the last $6$ rows of ${\bf F}'$ (virtual demands) satisfying~\eqref{eq:equations for example_2}. 
This equation is solvable because ${\bf E}^{\text{\rm T}}$ has dimension $18 \times 6$ and is of full rank; thus the left null space of ${\bf E}^{\text{\rm T}}$ contains $18-6 = 12$ linearly independent vectors ($6$ of them are the real demands), which could be exactly the rows of ${\bf F}'$. The explicit form of the matrix ${\bf F}'$ is provided in~\eqref{eq:example F'_2}. 
  
Finally, as described above, after determining  ${\bf F}'$, we can then determine $\sv^{n,1}$ and $\sv^{n,2}$  for each worker $n\in [6]$. The realizations of  ${\bf S}$ is given in~\eqref{eq:example S_2} and it can be checked that ${\bf S}$ is of full rank, satisfying Condition (c2).

\paragraph*{Decoding phase}
 Since ${\bf S}$ is of full rank, by 
 receiving ${\bf S} {\bf F}'  {\bf W}$, the master can multiple it by the first $6$ rows of the inverse of ${\bf S}$ to recover the computation task. 

In conclusion, the above computing scheme is decodable and optimal under the constraint of
the cyclic assignment. 
\hfill $\square$ 
\end{example}

We are now ready to generalize the proposed scheme in the above two examples. 
Note that when $\Ksf_{\rm c}\in [\frac{\Ksf}{\Nsf}\usf+1: \frac{\Ksf}{\Nsf}(\usf+1)-1]$ for some $\usf\in [\Nsf_{\rm r}-\msf+1]$, we add $\frac{\Ksf}{\Nsf}(\usf+1)-\Ksf_{\rm c}$ virtual demands and thus the communication cost 
is the same as the proposed scheme for the $(\Ksf,\Nsf,\Nsf_{\rm r}, \frac{\Ksf}{\Nsf}(\usf+1), \msf)$ distributed linearly separable computation problem, coinciding with the required communication cost in~\eqref{eq:achie case 2}. Hence, in the following we only consider the case   $\Ksf_{\rm c}=\frac{\Ksf}{\Nsf}\usf$ for each $\usf\in [\Nsf_{\rm r}-\msf+1]$.

 \paragraph{Data assignment phase} Datasets are assigned to the workers under the cyclic assignment, and each worker $n \in [\Nsf]$ has datasets with the indices in $\Zc_{n}$ given in~\eqref{eq:cyclic assignment n divides k}. 
 \paragraph{Computing phase} From the converse bound in  Theorem~\ref{thm:converse}, when $\usf \in [\Nsf_{\rm r} - \msf + 1]$, we have   $\Rsf^{*}_{cyc} \ge \frac{{\Nsf_{\rm r}\usf\Ksf}}{\Nsf(\msf + \usf - 1)}$. 
 So we divide each  message  $W_{k}$ where $ k \in [\Ksf]$   into $\dsf := \msf + \usf - 1$ non-overlapping and equal length sub-messages, and then the task function becomes $\Ksf_{\rm c} \dsf  $ linear combinations of these sub-messages. To coincide with the converse bound, each worker transmits $\frac{\Ksf}{\Nsf}\usf$ linear combinations of sub-messages. 
 We generate $\frac{\Ksf}{\Nsf}\usf \Nsf_{\rm r}-\frac{\Ksf}{\Nsf}\usf\dsf$ virtual demanded linear combinations of sub-messages, such that the  
  effective demand can be written as ${\bf F}'{\bf W}$, where ${\bf W} = [
W_{1,1};\ldots;W_{\Ksf,1};W_{1,2};\ldots;W_{\Ksf,\dsf}]$ and ${\bf F}'$ with dimension $\frac{\Ksf}{\Nsf}\usf\Nsf_{\rm r} \times \Ksf\dsf$ is with the form

\begin{align}
\left[\begin{array}{c:c:c:c}
 ({\bf F})_{\frac{\Ksf}{\Nsf}\usf \times \Ksf}  & {\bf 0}_{\frac{\Ksf}{\Nsf}\usf \times \Ksf}  & \cdots & {\bf 0}_{\frac{\Ksf}{\Nsf}\usf \times \Ksf}   \\ \hdashline
{\bf 0}_{\frac{\Ksf}{\Nsf}\usf \times \Ksf} &  ({\bf F})_{\frac{\Ksf}{\Nsf}\usf \times \Ksf}   & \cdots & {\bf 0}_{\frac{\Ksf}{\Nsf}\usf \times \Ksf}   \\ \hdashline 
 \vdots   & \vdots  &  \ddots& \vdots \\ \hdashline
 {\bf 0}_{\frac{\Ksf}{\Nsf}\usf \times \Ksf} &   {\bf 0}_{\frac{\Ksf}{\Nsf}\usf \times \Ksf}    & \cdots &  ({\bf F})_{\frac{\Ksf}{\Nsf}\usf \times \Ksf} \\ \hdashline
 ({\bf V_ {1}})_{ \frac{\Ksf\usf}{\Nsf}(\Nsf_{\rm r} - \dsf) \times \Ksf}  &  ({\bf V_ {2}})_{ \frac{\Ksf\usf}{\Nsf}(\Nsf_{\rm r} - \dsf) \times \Ksf}  &   \cdots   &   ({\bf V_ {\dsf }})_{ \frac{\Ksf\usf}{\Nsf}(\Nsf_{\rm r} - \dsf) \times \Ksf} 
 \end{array}
\right], 
\label{eq:general form of F'}
\end{align}
 where   $({{\bf V}_ i})_{ \frac{\Ksf}{\Nsf}\usf(\Nsf_{\rm r} - \dsf) \times \Ksf}, i \in [\dsf]$ are the coefficients 
 of  virtual demands; i.e.,  the first $\Ksf_{\rm c}\dsf$ rows of  ${\bf F}'$ correspond to the real demands and the last $\frac{\Ksf}{\Nsf}\usf \Nsf_{\rm r}-\Ksf_{\rm c}\dsf$ rows of ${\bf F}'$ correspond to the virtual demands. 
The selection on the elements in the   last $\frac{\Ksf}{\Nsf}\usf \Nsf_{\rm r}-\Ksf_{\rm c}\dsf$ rows of ${\bf F}'$ is the key step in our proposed computing scheme, which will be clarified later. 

 We let   each worker $n\in [\Nsf]$ transmit $\Ksf_{\rm c}$ linear combinations of the rows in ${\bf F}'{\bf W}$; each transmitted linear combination can be written as 
${\sv}^{n,j}{\bf F}'{\bf W}$ where ${\sv}^{n,j},n \in [\Nsf], j \in [\frac{\Ksf}{\Nsf}\usf]$ are defined as the coding vectors with length $\frac{\Ksf}{\Nsf}\usf\Nsf_{\rm r}$. 
Considering all the transmissions by all workers,  the coding matrix ${\bf S}$ is defined as:
\begin{align}
\mathbf{S}=  \left[\sv^{1,1};\ldots;\sv^{1,\frac{\Ksf}{\Nsf}\usf};\sv^{2,1};\ldots;\sv^{\Nsf,\frac{\Ksf}{\Nsf}\usf} \right]_{\Ksf\usf \times \frac{\Ksf}{\Nsf}\usf\Nsf_{\rm r}}.
\end{align} 
Worker $n \in [\Nsf]$ can only compute the messages $W_k$ where  $k\in \Zc_n$, so that the codewords sent by worker $n$ will not contain sub-messages $W_{i,j}$ where $i\notin \Zc_n$ and $j\in [\dsf]$.  The sub-messages that worker $n$ cannot compute are treated as its interferences, and the column-wise sub-matrix of ${\bf F}'$ corresponding to the sub-messages worker $n$ cannot compute is denoted by ${\bf F}'(n)$,
\begin{equation}
  {\bf F}'(n) = \left[ {\begin{array}{*{20}{c}}
\overline{ {\bf F}_{n} }& \cdots &{\bf 0}\\
\vdots & \ddots & \vdots\\
{\bf 0} & \cdots &\overline{ {\bf F}_{n} }\\
\overline{{\bf V}_{1,n}}& \cdots &\overline{{\bf V}_{\dsf,n}}
\end{array}} \right]_{\frac{\Ksf}{\Nsf}\usf\Nsf_{\rm r} \times \frac{\Ksf}{\Nsf}(\Nsf_{\rm r} - \msf)\dsf}, \label{F'n general}   
\end{equation}
where we recall that $\overline{ {\bf F}_{n}}:={\bf F}^{\left(\overline{  \Zc_{n} }  \right)_{\rm c}}$ is the column-wise sub-matrix of ${\bf F}$ containing the columns in $\overline{  \Zc_{n} } $, and $\overline{{\bf V}_{i,n}}, i \in [\dsf]$ is the column-wise sub-matrix of ${\bf V}_{i}$ containing the columns in $\overline{  \Zc_{n} } $.  
Thus each ${\sv}^{n,j}$ where $j \in [\frac{\Ksf\usf}{\Nsf}]$ should be a left null vector of ${\bf F}'(n)$. 

 \paragraph{Decoding phase} For any set of workers   $\Ac\subseteq [\Nsf]$ where $|\Ac|= \Nsf_{\rm r}$, the  matrix including the  coding vectors of the workers in $\Ac$ is denoted by
  $\mathbf{S}^{\Ac}$ (recall that  $\Ac(i)$ represents the $i^{\text{th}}$ smallest element of $\Ac$), where   
\begin{align}
\mathbf{S}^{\Ac}= \left[\sv^{\Ac(1),1};\ldots;\sv^{\Ac(1),\frac{\Ksf}{\Nsf}\usf};\sv^{\Ac(2),1};\ldots;\sv^{\Ac(\Nsf_{\rm r}),\frac{\Ksf}{\Nsf}\usf} \right]_{\frac{\Ksf}{\Nsf}\usf\Nsf_{\rm r} \times \frac{\Ksf}{\Nsf}\usf\Nsf_{\rm r}}. 
\end{align}

We need to guarantee that matrix ${\bf S}^{\Ac}$ is a non-singular matrix  so that 
when the master receives $\mathbf{S}^{\Ac}{\bf F}'{\bf W}$, the computation task can be recovered through computing $(\mathbf{S}^{\Ac})^{-1}\mathbf{S}^{\Ac}{\bf F}'{\bf W}$ .
 
\paragraph{Selection on the virtual demands in ${\bf F}'$}  Next, we will describe the core of our computing scheme, the selection on   the elements in the   last $\frac{\Ksf}{\Nsf}\usf (\Nsf_{\rm r}-\dsf)$ rows of ${\bf F}'$. Since the resulting  ${\bf F}'$ should satisfy the following two conditions for the sake of successful encoding and decoding:  
\begin{itemize}
\item (c1) For each $n\in [\Nsf]$, the rank of ${\bf F}'(n)$ should be $\frac{\Ksf\usf}{\Nsf}(\Nsf_{\rm r}-1)$, such that ${\bf F}'(n)$ contains $\frac{\Ksf\usf}{\Nsf}$ linearly independent left null vectors; let $\sv^{n,1}, \ldots,  \sv^{n,\frac{\Ksf\usf}{\Nsf}}$ be these $\frac{\Ksf\usf}{\Nsf}$ vectors.
\item (c2) For any  set $\Ac \subseteq [\Nsf], |\Ac|=\Nsf_{\rm r}$, the   matrix $\mathbf{S}^{\Ac}$ is of full rank.
\end{itemize}
 
We will consider three cases respectively:  $\Ksf_{\rm c}=\frac{\Ksf}{\Nsf}(\Nsf_{\rm r}-\msf+1)$, $\Ksf_{\rm c} = \frac{\Ksf}{\Nsf}(\Nsf_{\rm r}-\msf)$ and $\Ksf_{\rm c} \in \left(\frac{\Ksf}{\Nsf} :\frac{\Ksf}{\Nsf}(\Nsf_{\rm r} - \msf-1)\right]$. 
 
When $\Ksf_{\rm c} =\frac{\Ksf}{\Nsf}(\Nsf_{\rm r}-\msf+1)$, i.e., $\usf=\Nsf_{\rm r}-\msf+1$, we have $\dsf=\msf+\usf-1=\Nsf_{\rm r}$. Hence, 
the dimension of ${\bf F}'$ is  $\frac{\Ksf}{\Nsf}\usf\Nsf_{\rm r} \times \Ksf\dsf= \frac{\Ksf}{\Nsf}\usf \dsf \times \Ksf\dsf$.
Recall that the number of real demanded linear combinations of sub-messages is $\Ksf_{\rm c} \dsf= \frac{\Ksf}{\Nsf}\usf \dsf$; thus in this case,  ${\bf F}' $ does not contain virtual demands and  can be written as    
  \begin{align}
  {\bf F}'=
 \left[\begin{array}{c:c:c:c}
 ({\bf F})_{\frac{\Ksf}{\Nsf}\usf \times \Ksf}  & {\bf 0}_{\frac{\Ksf}{\Nsf}\usf \times \Ksf}  & \cdots & {\bf 0}_{\frac{\Ksf}{\Nsf}\usf \times \Ksf}   \\ \hdashline
{\bf 0}_{\frac{\Ksf}{\Nsf}\usf \times \Ksf} &  ({\bf F})_{\frac{\Ksf}{\Nsf}\usf \times \Ksf}   & \cdots & {\bf 0}_{\frac{\Ksf}{\Nsf}\usf \times \Ksf}   \\ \hdashline 
 \vdots   & \vdots  &  \ddots& \vdots \\ \hdashline
 {\bf 0}_{\frac{\Ksf}{\Nsf}\usf \times \Ksf} &   {\bf 0}_{\frac{\Ksf}{\Nsf}\usf \times \Ksf}    & \cdots &  ({\bf F})_{\frac{\Ksf}{\Nsf}\usf \times \Ksf}  
 \end{array}
\right]_{\frac{\Ksf}{\Nsf}\usf \dsf \times \Ksf\dsf}. 
\label{F(Kc)}
 \end{align}

For each worker  $n$ where $n\in [\Nsf]$, the number of messages it cannot compute is $\Ksf - \Msf= \frac{\Ksf}{\Nsf}(\usf-1)$. 
Note that  the dimension of $\overline{ {\bf F}_{n} }$ is $ \frac{\Ksf\usf}{\Nsf} \times \frac{\Ksf(\usf-1)}{\Nsf}$ and $\overline{ {\bf F}_{n} }$ is full-rank with high probability.
Hence, there are exactly $\frac{\Ksf}{\Nsf}$ linearly independent left null  vectors of $\overline{ {\bf F}_{n} }$, which are denoted by  ${\bf s}_{n,k}, k \in [\frac{\Ksf}{\Nsf}]$.
Then the coding vectors of worker $n$, denoted as $\{\sv^{n,(k-1)\usf + j}, \forall k \in [\frac{\Ksf}{\Nsf}], j \in [\usf]\}$, can be designed as
\begin{align*}
[m_{n,k,j,1}{\bf s}_{n,k}, m_{n,k,j,2}{\bf s}_{n,k},\ldots, m_{n,k,j,\dsf}{\bf s}_{n,k} ]_{1 \times \frac{\Ksf\usf}{\Nsf}\dsf},
\end{align*}
where each $m_{n,k,j,i}, i \in [\dsf]$   is selected uniformly i.i.d. over $\mathbb{F}_{\qsf}$. 
 Note that the coding vectors of worker $n$ are linearly independent with high probability since $\dsf \geq \usf$.

When $\Ksf_{\rm c} = \frac{\Ksf}{\Nsf}(\Nsf_{\rm r}-\msf) $, i.e., $\usf=\Nsf_{\rm r}-\msf$, the number of rows corresponding to virtual demands in ${\bf F}'$ is $\frac{\Ksf\usf}{\Nsf}\Nsf_{\rm r}-\Ksf_{\rm c} \dsf=\frac{\Ksf\usf}{\Nsf} $. Thus 
${\bf F}'$ can be written as
\begin{align}
 \mathbf{F}' =\left[\begin{array}{c:c:c:c}
 ({\bf F})_{\frac{\Ksf}{\Nsf}\usf \times \Ksf}  & {\bf 0}_{\frac{\Ksf}{\Nsf}\usf \times \Ksf}  & \cdots & {\bf 0}_{\frac{\Ksf}{\Nsf}\usf \times \Ksf}   \\ \hdashline
{\bf 0}_{\frac{\Ksf}{\Nsf}\usf \times \Ksf} &  ({\bf F})_{\frac{\Ksf}{\Nsf}\usf \times \Ksf}   & \cdots & {\bf 0}_{\frac{\Ksf}{\Nsf}\usf \times \Ksf}   \\ \hdashline 
 \vdots   & \vdots  &  \ddots& \vdots \\ \hdashline
 {\bf 0}_{\frac{\Ksf}{\Nsf}\usf \times \Ksf} &   {\bf 0}_{\frac{\Ksf}{\Nsf}\usf \times \Ksf}    & \cdots &  ({\bf F})_{\frac{\Ksf}{\Nsf}\usf \times \Ksf} \\ \hdashline
 ({\bf V_ {1}})_{\frac{\Ksf}{\Nsf}\usf \times \Ksf}  &  ({\bf V_ {2}})_{\frac{\Ksf}{\Nsf}\usf \times \Ksf}  &   \cdots   &   ({\bf V_ {\dsf }})_{\frac{\Ksf}{\Nsf}\usf \times \Ksf} 
 \end{array}
\right]_{\frac{\Ksf}{\Nsf}\usf\Nsf_{\rm r} \times \Ksf\dsf}
, \label{eq:second case F'}
\end{align} 
 where ${\bf V}_ i, i \in [\dsf]$ are the coefficients of virtual demands. In this case, we select each element in  ${\bf V}_ i, i \in [\dsf]$ 
uniformly i.i.d. over $\mathbb{F}_{\qsf}$. For each worker $n\in [\Nsf]$, the    dimension of ${\bf F}'(n)$ is $\frac{\Ksf}{\Nsf}\usf\dsf  \times (\Ksf - \Msf)(\Nsf_{\rm r} - 1) =  \frac{\Ksf}{\Nsf}\usf\Nsf_{\rm r} \times \frac{\Ksf}{\Nsf}\usf(\Nsf_{\rm r} - 1) $. Obviously, ${\bf F}'(n)$ is of full rank with high probability. Thus 
the left null space of ${\bf F}'(n)$ contains  exactly  $\frac{\Ksf}{\Ksf}\usf$ left null vectors with high probability. 
We let the  $\frac{\Ksf}{\Nsf}\usf$ coding vectors of worker $n$ be these  $\frac{\Ksf}{\Nsf}\usf$ left null vectors.

Finally, we consider the case $\Ksf_{\rm c} = \frac{\Ksf}{\Nsf}\usf $ where $\usf \in [\Nsf_{\rm r} - \msf-1]$; in this case  ${\bf F}'$ and $\mathbf{F}'(n)$ are shown as in~\eqref{eq:general form of F'} and~\eqref{F'n general}, respectively. 
Since $\usf \in [\Nsf_{\rm r} - \msf-1]$,  we have (recall that the case $\msf=1$ has been solved in~\cite{m=1} and thus here we consider $\msf>1$)
\begin{subequations}
\begin{align}
&(\Nsf_{\rm r}-\msf)(\msf-1)>\usf(\msf-1)\\
&\Longrightarrow  (\Nsf_{\rm r}-\msf)(\msf+\usf-1)>\usf (\Nsf_{\rm r}-1) \\
&\Longrightarrow \frac{\Ksf\usf}{\Nsf}\Nsf_{\rm r} - \frac{\Ksf}{\Nsf}(\Nsf_{\rm r} - \msf)\dsf < \frac{\Ksf}{\Nsf}\usf. \label{eq:not enough space}
\end{align}
\end{subequations}
If we select the coefficients of virtual demands uniformly i.i.d. over $\mathbb{F}_{\qsf}$, ${\bf F}'(n)$ with dimension $\frac{\Ksf}{\Ksf}\usf\Nsf_{\rm r} \times \frac{\Ksf}{\Nsf}(\Nsf_{\rm r} - \msf)\dsf$ is of full rank with high probability  for each $n\in [\Nsf ]$; thus by~\eqref{eq:not enough space}, ${\bf F}'(n)$ cannot contain $\frac{\Ksf}{\Nsf}\usf$ linearly independent left null vectors, which contradicts with Condition (c1). 
Hence, we need to design the coefficients of the virtual demands (i.e., the last $\frac{\Ksf}{\Nsf}\usf \Nsf_{\rm r}-\Ksf_{\rm c}\dsf$ rows of ${\bf F}'$) satisfying Condition (c1); thus 
 there exist $ \frac{\Ksf}{\Nsf}(( \Nsf_{\rm r} - \msf )\dsf - \Nsf_{\rm r}\usf + \usf) = \frac{\Ksf}{\Nsf}(\msf - 1)(\Nsf_{\rm r} - \msf - \usf )$ linearly independent combinations of the columns in ${\bf F}'(n)$  where $n\in [\Nsf]$.  

More precisely, each sub-message that one worker cannot compute can be treated as an `interference' to this worker and the number of sub-messages that  each worker cannot compute is $\frac{\Ksf}{\Nsf}(\Nsf_{\rm r} - \msf)\dsf$.
 For each worker $n\in [\Nsf]$, we should align its interferences such that the dimension of its interferences (i.e., the rank of ${\bf F}'(n)$) is reduced to  $\frac{\Ksf}{\Nsf}( \Nsf_{\rm r}\usf - \usf)$. 
In the proposed computing scheme based on interference alignment, we generate linear equations with the form 
 \begin{align}
 {\bf F}'{\ev^{\text{\rm T}}}  = {\bf 0}_{\frac{\Ksf}{\Ksf}\usf\Nsf_{\rm r} \times 1},\label{eq:en useful}
 \end{align}
where $\ev$ with dimension $1 \times \Ksf\dsf $ represents one  linear reduction vector.  The interference alignment strategy is based on the observation that under the cyclic assignment, each   adjacent $\Nsf_{\rm r} - \msf - \usf$ workers  in a cyclic wrap-around fashion cannot compute a common set of  $\frac{\Ksf}{\Nsf}(\usf+1)$ messages. 
For each $i\in [\Nsf]$, all workers in 
$$
\Nc_i:=\{i, \text{Mod}(i+1,\Nsf), \ldots, \text{Mod}(i+\Nsf_{\rm r} - \msf - \usf-1,\Nsf) \}
$$
do not have $D_{(t-1)\Nsf+j}$ where $j\in \Kc_i:= \{ \text{Mod}(i-1,\Nsf), \text{Mod}(i-2,\Nsf),\ldots, \text{Mod}(i-\usf-1,\Nsf) \} $ and $t\in \left[\frac{\Ksf}{\Nsf} \right]$. In other words,  for each $i\in [\Nsf]$,  each dataset in 
\begin{align*}
\Qc_i:=\cup_{t\in \left[\frac{\Ksf}{\Nsf} \right]} \left( \Kc_i + (t-1)\Nsf \right) 
\end{align*}
 is not assigned to any workers in $\Nc_i$; thus  the sub-messages $W_{k,j}$ where $k\in \Qc_i$ and $j\in [\dsf]$ are treated as interferences to each worker in $\Nc_i$.    
For the sake of interference alignment, we will generate $\frac{\Ksf}{\Nsf} (\msf-1)$ linear reduction vectors  $\ev_{i,t,v}$ where $t\in \left[\frac{\Ksf}{\Nsf} \right]$ and $v\in [\msf-1]$,  simultaneously for the workers in $\Nc_i$. Note that each vector $\ev_{i,t,v}$  leads to one linear equation ${\bf F}' \ev_{i,t,v}^{\text{\rm T}}={\bf 0}_{\frac{\Ksf\usf}{\Nsf}\Nsf_{\rm r} \times 1}$ and its elements with indices in $[\Ksf \dsf] \setminus \underset{j\in [\dsf]}{\cup}  \left( \Qc_{i} + (j-1)\Ksf \right) $ are all $0$.  
  In other words, in $\ev_{i,t,v}$, only the elements with indices in $\underset{j\in [\dsf]}{\cup}  \left( \Qc_{i} + (j-1)\Ksf \right)$ could be non-zero.
  

The construction on each $\ev_{i,t,v}$ is as follows. Since the elements in ${\bf F}$ are generated i.i.d over $\mathbb{F}_{\qsf}$, ${\bf F}^{(\Qc_i)_c}$ is full-rank with high probability. 
In addition, the dimension of   ${\bf F}^{(\Qc_i)_c}$ is  $\frac{\Ksf}{\Nsf}\usf \times \frac{\Ksf}{\Nsf}(\usf+1)$. 
Thus with high probability the right null space of ${\bf F}^{(\Qc_i)_c}$ contains $\frac{\Ksf}{\Nsf}$ linearly independent vectors, denoted by ${\bf y}_{i, 1}, \ldots, {\bf y}_{i, \frac{\Ksf}{\Nsf}}$. 
By the structure of ${\bf F}'$ as illustrated in~\eqref{eq:general form of F'}, for each $j \in [\dsf]$, we can let the sub-vector of $\ev_{i,t,v}$ including the elements  with indices in $\Qc_{i} + (j -1)\Ksf$    be a  multiply of ${\bf y}_{i,t}$, such that   ${\bf F}' \ev_{i,t,v}^{\text{\rm T}}={\bf 0}_{\frac{\Ksf\usf}{\Nsf}\Nsf_{\rm r} \times 1}$ can hold.
Note that, we select each multiply number uniformly i.i.d. over $\mathbb{F}_{\qsf}$.

Under the construction above considering all $i\in [\Nsf]$,
we can see that  each worker $k\in [\Ksf]$ is contained in each of the $(\Nsf_{\rm r} - \msf - \usf )$ sets  $\Nc_{k}, \Nc_{\text{Mod}(k-1,\Nsf)},\ldots, \Nc_{\text{Mod}(k-\Nsf_{\rm r} + \msf + \usf+1,\Nsf)} $, and for each of those worker sets we generate $\frac{\Ksf}{\Nsf} (\msf-1)$ linear reduction vectors simultaneously for all workers in the set. Thus we totally generate $\frac{\Ksf}{\Nsf} (\msf-1) (\Nsf_{\rm r} - \msf - \usf )$ linear reduction vectors for each worker. In Appendix~\ref{sec: proof of linear independence for each worker}, we prove the following lemma.
\begin{lem}
\label{lem:linearly independent for each worker}
For each worker $n \in [\Nsf]$, the generated $\frac{\Ksf}{\Nsf} (\msf-1) (\Nsf_{\rm r} - \msf - \usf )$ linear reduction vectors are linearly independent with high probability.
\hfill $\square$
\end{lem}
By Lemma~\ref{lem:linearly independent for each worker}, Condition (c1) is satisfied.

Next, we define a  matrix $\bf E$   as follows,
\begin{equation}
   {\bf E} = \left[ \ev_{1,1,1};\ldots;\ev_{1,1,\msf -1};\ev_{1,2,1};\ldots;\ev_{1,\frac{\Ksf}{\Nsf},\msf -1};\ldots;\ev_{\Nsf,\frac{\Ksf}{\Nsf},\msf -1} \right],
\end{equation}
with dimension $\Ksf(\msf -1) \times \Ksf\dsf$.


The next step is to solve the virtual demands in ${\bf F}'$ (${\bf V}_ i, \forall i \in [\dsf]$) satisfying,
\begin{equation}
{\bf F}'{\bf E}^{\text{\rm T}} = {\bf 0}_{\frac{\Ksf}{\Nsf}\usf\Nsf_{\rm r} \times  \Ksf\left( \msf -1 \right)}.
\end{equation} 
Since ${\bf E}^{\text{\rm T}}$ is of full rank with dimension $\Ksf\dsf \times  \Ksf\left( \msf -1 \right)$, there are  $\Ksf\dsf -  \Ksf\left( \msf -1 \right) = \Ksf\usf$  independent left null vectors  of ${\bf E}^{\text{\rm T}}$. The first $\frac{\Ksf}{\Nsf}\usf\dsf$ rows in ${\bf F}'$ are  from real demands, and lie in the left null space of ${\bf E}^{\text{\rm T}}$,   by our construction on ${\bf E}^{\text{\rm T}}$. Hence, 
 the   remaining rows   in ${\bf F}'$  (corresponding to  virtual demands) can be designed as  the remaining $\frac{\Ksf}{\Nsf}\usf(\Nsf - \dsf )$  left null vectors  of ${\bf E}^{\text{\rm T}}$. 

Finally, the coding matrix ${\bf S}$ can be obtained  from Condition (c1), where we let  $${{\bf s}^{n,j}}{\bf F}'(n)= {\bf 0}_{1 \times (\Nsf_{\rm r} - \msf)\dsf}, n \in \left[ \Nsf \right],j \in \left[ \frac{\Ksf}{\Nsf}\usf \right]. $$  The decodability proof on the proposed scheme is provided in Appendix~\ref{sec:decodability}, which is based on the Schwartz-Zippel lemma~\cite{Schwartz, Zippel, Demillo_Lipton}.
As a result, the communication cost by the proposed computing scheme is $\frac{\Nsf_{\rm r} \Ksf \usf}{\Nsf(\msf+\usf-1)} $, as given in~\eqref{eq:achie case 2}.

\begin{figure}
    \centering
    \begin{subfigure}[t]{\linewidth}
        \centering
        \includegraphics[scale=0.35]{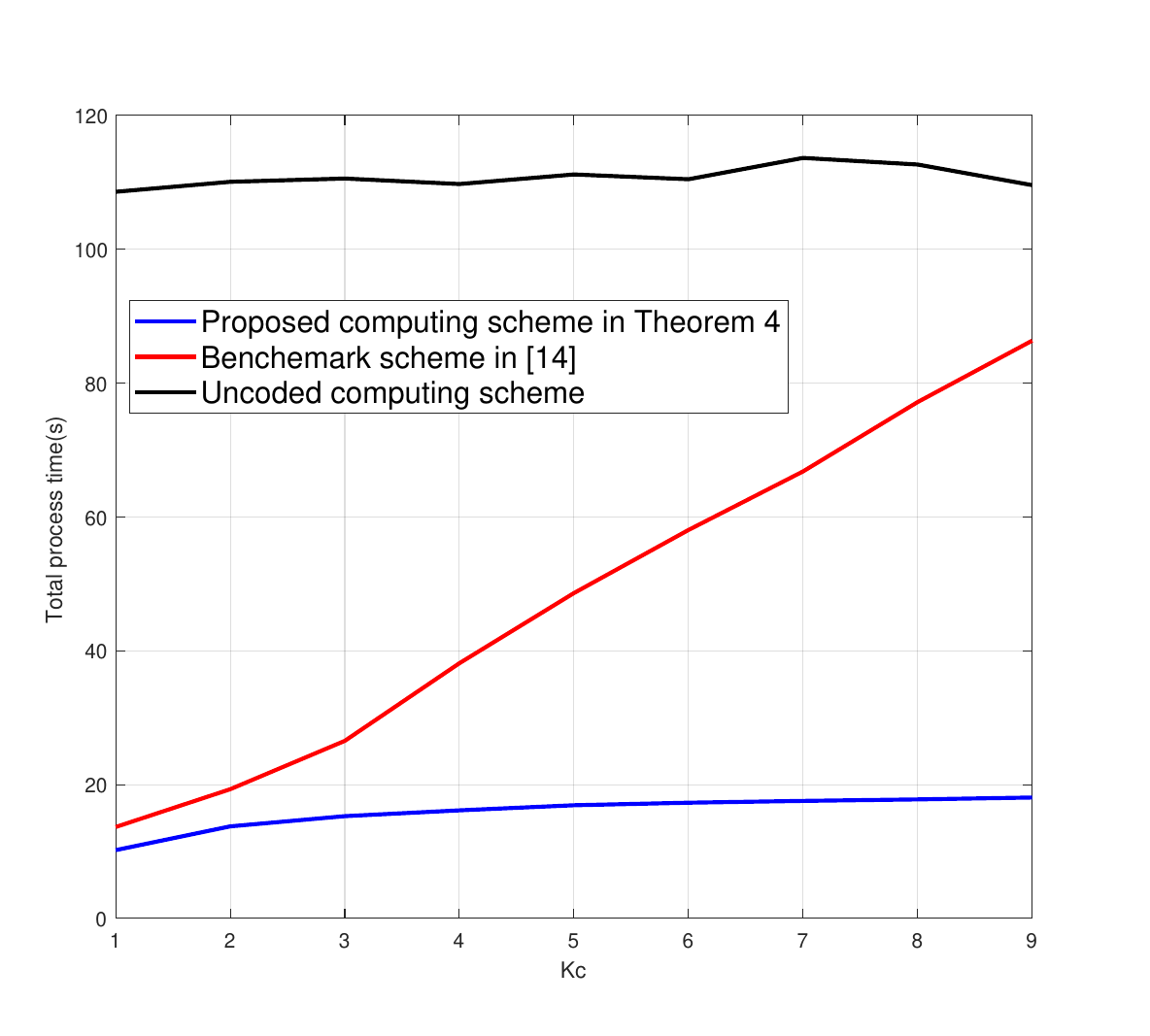}
        \caption{\small $(\Ksf=\Nsf=\Nsf_{\rm r}=12, \Ksf_{\rm c}\in [9], \msf = 3)$}
        \label{fig:m3Kc}
    \end{subfigure}
    \vskip 0.5em
    \begin{subfigure}[t]{\linewidth}
        \centering
        \includegraphics[scale=0.35]{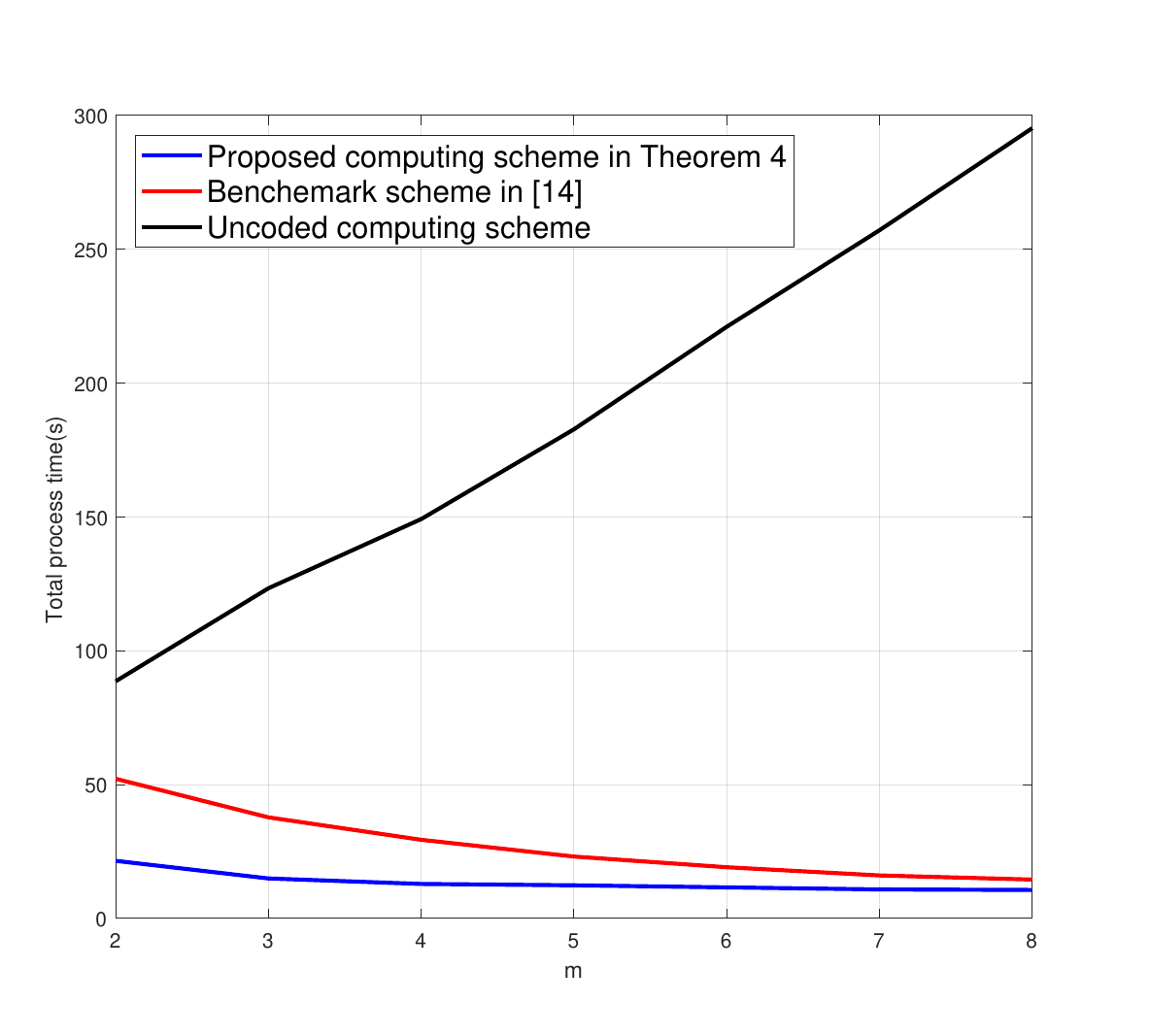}
        \caption{\small $(\Ksf=\Nsf=\Nsf_{\rm r}=12, \Ksf_{\rm c}= 4, \msf \in [2:8])$}
        \label{fig:Kc4m}
    \end{subfigure}
    \caption{\small System computation time of the proposed scheme, the benchmark scheme and the uncoded scheme for $\Ksf_{\rm c} \in (\frac{\Ksf}{\Nsf}:\frac{\Ksf}{\Nsf}(\Nsf_{\rm r} - \msf +1)]$.}
    \label{fig:side-by-side}
\end{figure}
\section{Experimental Results}
We implement our proposed coded computing scheme in Python3.11 by using the MPI4py package over the Tencent Cloud.
As a benchmark scheme, we repeat the computing scheme in~\cite{CCG} by $\Ksf_{\rm c}$ times, where each time by the   computing scheme in~\cite{CCG} the master can recover one linear combination of messages. 
 While the datasets are assigned offline, the experiment focuses on  the distributed computing process time (including computing, encoding, transmission and decoding times) .
\paragraph{Tencent Cloud Setup} We choose Tencent Cloud instances, specifically S6.LARGE16 and S6.MEDIUM4. We choose S6.LARGE16 as the master and S6.MEDIUM4 as the workers. 
These Tencent Cloud instances are equipped with Intel Xeon Ice Lake processors running at a base clock speed of 2.7 GHz with a turbo frequency of 3.3 GHz. All instances used in our experiment are identical in terms of computing power, memory and network resources. The communication speed between the server and users is 1MB/s. To generate the demand matrix and input vectors, we set the field size $\qsf$ to $11$ and the elements in the demand matrix $\mathbb{F}$ are  uniformly and independently chosen  over $\mathbb{F}_{11}$.  In the system configuration represented by $(\Nsf,\Ksf,\Ksf_{\rm c},\msf),$ we employ Monte Carlo methods with $100$ samples and then average the resulting time over these $100$ samples.
\paragraph{Proposed scheme v.s. Benchmark scheme} We compare the distributed computing process times of our proposed scheme, the benchmark scheme and the uncoded scheme that workers transmit uncodewords to the master. 
It can be seen from Fig.~\ref{fig:m3Kc} that, the proposed scheme outperforms than the Benchmark scheme  under the system parameters $(\Ksf=\Nsf=\Nsf_{\rm r}=12, \Ksf_{\rm c}\in [9], \msf = 3)$, where the reduction percentage ranges from $58.5\%$ to $72.8\%$.  From Fig.~\ref{fig:Kc4m} under the system parameters $(\Ksf=\Nsf=\Nsf_{\rm r}=12, \Ksf_{\rm c}= 4, \msf \in [2:8])$,  the reduction percentage of the process time of our scheme compared to the benchmark scheme ranges from $57.1\%$ to $59.7\%$. The experimental results  align with the theoretical perspective that the communication cost of the proposed scheme is significantly less than that of the benchmark scheme, given the same computation cost.

\section{Conclusion}
\label{sec:Conclusion}
In this paper, we considered the fundamental limits of distributed linearly separable computation under cyclic assignment. We proposed a coded computing scheme using interference alignment  which can achieve the (order) optimal communication cost under the cyclic assignment. The proposed scheme will effectively decrease the communication cost by the computing scheme under the repetition assignment. Experimental results over Tencent Cloud showed the improvement of our proposed scheme compared to the benchmark scheme.
On-going works include  the improved computing schemes with combinatorial assignment and the distributed linearly separable computation problem with heterogeneous computation capabilities.

  \appendices

\section{The explicit form of matrices in the example}

\begin{figure*}[t]
\scriptsize
\renewcommand{\arraystretch}{0.8}
\begin{align}
{\bf E} =
\begin{bmatrix}
\ev_{1}\\
\vdots\\
\ev_{6}
\end{bmatrix}
=
\left[
\begin{array}{@{\hskip 1pt}cccccccccccccccccc@{\hskip 1pt}}
 0 & 0 & 0 & 1 & -2 & 1 & 0 & 0 & 0 & 1 & -2 & 1 & 0 & 0 & 0 & 2 & -4 & 2 \\
 0 & 0 & 0 & 0 & 0 & 0 & 1 & 0 & 0 & 0 & -5 & 4 & 0 & 0 & 0 & 0 & 0 & 0 \\
 8 & -10 & 0 & 0 & 0 & 2 & 4 & -5 & 0 & 0 & 0 & 1 & 0 & 0 & 0 & 0 & 0 & 0 \\
 2 & -4 & 2 & 0 & 0 & 0 & 1 & -2 & 1 & 0 & 0 & 0 & 2 & -4 & 2 & 0 & 0 & 0 \\
 0 & 1 & -2 & 1 & 0 & 0 & 0 & 2 & -4 & 2 & 0 & 0 & 0 & 2 & -4 & 2 & 0 & 0 \\
 0 & 0 & 2 & -4 & 2 & 0 & 0 & 0 & 1 & -2 & 1 & 0 & 0 & 0 & 0 & 0 & 0 & 0 \\
\end{array}
\right].
\label{eq:E for example_2}
\end{align}
\end{figure*}

\begin{figure*}
\scriptsize
\renewcommand{\arraystretch}{0.8}
\begin{equation}
\small
{\bf F}' = \left[
\begin{array}{@{\hskip 1pt}cccccccccccccccccc@{\hskip 1pt}}
 1 & 1 & 1 & 1 & 1 & 1 & 0 & 0 & 0 & 0 & 0 & 0 & 0 & 0 & 0 & 0 & 0 & 0 \\
 1 & 2 & 3 & 4 & 5 & 6 & 0 & 0 & 0 & 0 & 0 & 0 & 0 & 0 & 0 & 0 & 0 & 0 \\
 0 & 0 & 0 & 0 & 0 & 0 & 1 & 1 & 1 & 1 & 1 & 1 & 0 & 0 & 0 & 0 & 0 & 0 \\
 0 & 0 & 0 & 0 & 0 & 0 & 1 & 2 & 3 & 4 & 5 & 6 & 0 & 0 & 0 & 0 & 0 & 0 \\
 0 & 0 & 0 & 0 & 0 & 0 & 0 & 0 & 0 & 0 & 0 & 0 & 1 & 1 & 1 & 1 & 1 & 1 \\
 0 & 0 & 0 & 0 & 0 & 0 & 0 & 0 & 0 & 0 & 0 & 0 & 1 & 2 & 3 & 4 & 5 & 6 \\
 -159 & -117 & -104 & -51 & 2 & 1 & 18 & -2 & 0 & 9 & 18 & 18 & 0 & 0 & 18 & 9 & 0 & 18 \\
 -264 & -159 & -123 & -51 & 2 & 1 & 108 & -14 & 2 & 0 & 36 & 18 & 9 & 9 & 9 & 9 & 0 & 18 \\
 0 & 80 & 70 & 36 & 0 & 0 & 180 & -16 & 4 & 18 & 36 & 0 & 18 & 18 & 0 & 0 & 18 & 18 \\
 -170 & -72 & -49 & -14 & 2 & 0 & 144 & -11 & 2 & 0 & 36 & 9 & 18 & 18 & 9 & 9 & 9 & 18 \\
 -49 & 19 & 6 & -1 & 0 & 1 & 90 & -44 & 2 & 18 & 18 & 0 & 9 & 9 & 0 & 18 & 9 & 0 \\
 -88 & -24 & -26 & -16 & 2 & 2 & 90 & -20 & 2 & 18 & 18 & 0 & 18 & 18 & 18 & 18 & 9 & 9 \\
\end{array}
\right].
\label{eq:example F'_2}
\end{equation} 
\end{figure*}

\begin{figure*}
\scriptsize
\renewcommand{\arraystretch}{0.8}
\begin{equation}
\small
   {\bf S} = \left[
\begin{array}{@{\hskip 1pt}cccccccccccccccccc@{\hskip 1pt}}
 -150 & 26 & -9 & 9 & -297 & 45 & 5 & -17 & -7 & 19 & 0 & 3 \\
 -12978 & 2294 & 1746 & -18 & -21438 & 2988 & 632 & -1568 & -559 & 1690 & 150 & 0 \\
 15 & -3 & -108 & 18 & -9 & 0 & 0 & 1 & 1 & -2 & 0 & 1 \\
 1016 & -188 & -4716 & 720 & 306 & -126 & 18 & 64 & 63 & -120 & 30 & 0 \\
 120 & -24 & 216 & -36 & 18 & 0 & 0 & 8 & 3 & -16 & 0 & 8 \\
 -715 & 125 & -423 & 63 & -414 & 0 & 45 & -85 & -22 & 85 & 5 & 0 \\
 -55980 & 9346 & -2394 & 774 & -612 & 0 & 150 & -282 & 143 & 14 & 0 & 18 \\
 -35991 & 6013 & -1449 & 477 & -396 & 0 & 105 & -201 & 83 & 35 & 9 & 0 \\
 -620 & 104 & -306 & 126 & -738 & 0 & 0 & -68 & -43 & 86 & 0 & 32 \\
 -14196 & 1432 & -14046 & 5298 & -10926 & 0 & 384 & -1692 & -781 & 2138 & 192 & 0 \\
 -11140 & 1988 & 10368 & -3528 & -13536 & 1350 & 150 & -646 & -171 & 792 & 0 & 304 \\
 -93492 & 18036 & 34416 & -11160 & -29808 & 990 & 1326 & -3230 & -703 & 3560 & 608 & 0 \\
\end{array}
\right].\label{eq:example S_2}
\end{equation} 
\end{figure*}

\section{Proof of Corollary~\ref{thm: Compared} }
\label{sec: proof of compared}
When $\Ksf = \Nsf,$ we have $\Rsf^{\star}_{\text{cyc}} \geq \frac{\Nsf_{\rm r}\Ksf_{\rm c}}{\msf + \Ksf_{\rm c} -1}, ~\text{if}~ \Ksf_{\rm c} \in \left[ \Nsf_{\rm r} - \msf + 1 \right]$; and we have $\Rsf^{\star}_{\text{cyc}} \geq \Ksf_{\rm c}, ~\text{if}~ \Ksf_{\rm c} \in \left[ \Nsf_{\rm r} - \msf + 1 : \Ksf \right]$ from Theorem~\ref{thm:converse}.
This converse bound can be directly achieved by the proposed scheme in Theorem~\ref{thm:main achievable scheme}.


We then consider the case where $\Nsf$ divides $\Ksf$.

When $\Ksf_{\rm c} \in \left[\frac{\Ksf}{\Nsf} \right],$ we have $\Rsf^{\star}_{\text{cyc}} \geq \frac{\Ksf_{\rm c} \Nsf_{\rm r}}{\msf}$ from Theorem~\ref{thm:converse}, which can be directly achieved by the proposed scheme in Theorem~\ref{thm:main achievable scheme}. 

When $\Ksf_{\rm c} \in \left( \frac{\Ksf}{\Nsf}: \frac{\Ksf}{\Nsf}(\Nsf_{\rm r}-\msf+1) \right],$ we have $\Rsf^{\star}_{\text{cyc}} \geq \frac{\Nsf_{\rm r}\Ksf_{\rm c}}{\msf + \Ksf_{\rm c} -1}$ from Theorem~\ref{thm:converse}, while the proposed scheme in  Theorem~\ref{thm:main achievable scheme} can achieve the communication cost $\Rsf_{\text{ach}} = \frac{\Ksf\Nsf_{\rm r}\usf}{\Nsf(\msf + \Ksf_{\rm c} -1)}$, where $\usf := \lceil \frac{\Ksf_{\rm c}\Nsf}{\Ksf} \rceil$. 
Comparing the converse bound and the achievable communication cost, we find that $\frac{\Ksf\Nsf_{\rm r}\usf}{\Nsf(\msf + \Ksf_{\rm c} -1)} / \frac{\Nsf_{\rm r}\Ksf_{\rm c}}{\msf + \Ksf_{\rm c} -1} = \frac{\Ksf\usf}{\Nsf\Ksf_{\rm c}} < 2$.  The proposed scheme achieves the converse bound within a factor of 2.

When $\Ksf_{\rm c} \in \left( \frac{\Ksf}{\Nsf}(\Nsf_{\rm r}-\msf+1) : \Ksf \right],$ we have $\Rsf^{\star} = \Rsf^{\star}_{\text{cyc}} \geq \Ksf_{\rm c}$ from Theorem~\ref{thm:converse}, which can be achieved by  the proposed scheme in  Theorem~\ref{thm:main achievable scheme}. 

In conclusion, we proved Corollary~\ref{thm: Compared}.

\section{Proof of Lemma~\ref{lem:linearly independent for each worker}}
\label{sec: proof of linear independence for each worker}  
By the symmetry, we only need to prove  that the  $\frac{\Ksf}{\Nsf} (\msf-1) (\Nsf_{\rm r} - \msf - \usf )$ linear reduction vectors
  for  worker $1$ are linearly independent. Note that among $\Nsf_{1}, \ldots, \Nc_{\Nsf} $, worker $1$  is contained in $\Nc_{1},\Nc_{\Nsf}, \Nc_{\Nsf-1},\ldots, \Nc_{\Nsf-\Nsf_{\rm r}+\msf+\usf+2}$, totally $\Nsf_{\rm r} - \msf - \usf $ sets. 
  
By the construction, the $\frac{\Ksf}{\Nsf} (\msf-1)$ linear reduction vectors for the workers in $\Nsf_{1}$, denoted by  $\ev_{1,1,1},\ldots,\ev_{1,\frac{\Ksf}{\Nsf},\msf-1}$, are linearly independent with high probability. In each of these $\frac{\Ksf}{\Nsf} (\msf-1)$ vectors, only the elements with indices in 
\begin{align}
\bigcup_{j\in [\dsf]} \left(\bigcup_{t\in \left[ \frac{\Ksf}{\Nsf}\right]} \left(\{\Nsf,\ldots, \Nsf-\usf+2\} + (t-1) \Nsf \right) +(j-1)\Ksf \right)  \label{eq:Nc N non 0}
\end{align}
 could be non-zero.

 Then we consider  the $\frac{\Ksf}{\Nsf} (\msf-1)$ linear reduction vectors generated for the workers $\Nc_{\Nsf-1}$, denoted by  $\ev_{\Nsf,1,1},\ldots,\ev_{\Nsf,\frac{\Ksf}{\Nsf},\msf-1}$. 
  In each of these $\frac{\Ksf}{\Nsf} (\msf-1)$ vectors, only the elements with indices in 
\begin{align}
\bigcup_{j\in [\dsf]} \left(\bigcup_{t\in \left[ \frac{\Ksf}{\Nsf}\right]} \left(\{\Nsf-1,\ldots, \Nsf-\usf+1\} + (t-1) \Nsf \right)+(j-1)\Ksf  \right)  \label{eq:Nc N-1 non 0}
\end{align}
 could be non-zero. It can be seen that each element  in  
\begin{align}
 \bigcup_{j\in [\dsf]} \left(\bigcup_{t\in \left[ \frac{\Ksf}{\Nsf}\right]} \left(\{\Nsf-\usf+1\} + (t-1) \Nsf \right) +(j-1)\Ksf  \right) \label{eq:Nc in N-1 not in N}
 \end{align} 
 is contained by the set in~\eqref{eq:Nc N-1 non 0} and not contained by the set in~\eqref{eq:Nc N non 0}.
 By the construction,  the column-wise sub-matrix of 
 \begin{align}
\left[ \ev_{\Nsf,1,1};\ldots;\ev_{\Nsf,\frac{\Ksf}{\Nsf},\msf-1}  \right]
\end{align}   including the columns with indices in~\eqref{eq:Nc in N-1 not in N}  has rank equal to $\frac{\Ksf}{\Nsf} (\msf-1)$ with high probability. Thus, the matrix 
\begin{align} \left[ \ev_{1,1,1};\ldots;\ev_{1,\frac{\Ksf}{\Nsf},\msf-1};\ev_{\Nsf,1,1};\ldots;\ev_{\Nsf,\frac{\Ksf}{\Nsf},\msf-1}  \right]
\end{align} 
is of full rank with high probability. 

Similarly we can prove that the matrix 
\begin{align}
\left[
\ev_{1,1,1},\, \ldots,\, \ev_{1,\frac{\Ksf}{\Nsf},\msf-1},\, 
\ev_{\Nsf,1,1},\, \ldots,\, \ev_{\Nsf,\frac{\Ksf}{\Nsf},\msf-1},\, \ldots, \right. \nonumber \\  \left.
\ev_{\Nsf-\Nsf_{\rm r}+\msf+\usf+2,1,1},\, \ldots,\, 
\ev_{\Nsf-\Nsf_{\rm r}+\msf+\usf+2,\frac{\Ksf}{\Nsf},\msf-1}
\right]
\end{align}
is of full rank, since $\Kc_{\Nsf-\Nsf_{\rm r}+\msf+\usf+2}=\{\Nsf-\Nsf_{\rm r}+\msf+\usf+1,\Nsf-\Nsf_{\rm r}+\msf+\usf,\ldots, \Nsf- \Nsf_{\rm r}+\msf+3 \}$, and $\Nsf_{\rm r}-\msf-3<\Nsf$.

\section{Decodability Proof of the Proposed Computing Scheme}
\label{sec:decodability}
\subsection{$\Ksf = \Nsf$}
In this section, we provide the sketch of the feasibility proof when $\Ksf = \Nsf$. 
To prove the decodability of the proposed computing scheme, we need to prove ${\bf S}^{\Ac}$ is full-rank. Assume there is  $\Ac = \{\Ac(1), \ldots, \Ac(\Nsf_{\rm r})\}$ and $\Ac(1) \textless \ldots \textless \Ac(\Nsf_{\rm r}).$ 
Rewrite ${\bf F}'$ into the basis of matrix ${\bf E}$, and consider the matrix ${\bf S}^{\Ac(i)}$, there is ${\bf S}^{\Ac(i)} {\bf F}'(\Ac(i)) = 0$. With the rank reduction matrix $\bf E$, the constraint can be reduced to ${\bf S}^{\Ac(i)} {{\bf F}'(\Ac(i))}^{(1,2,\ldots,\Ksf_{\rm c}(\Nsf_{\rm r}-1))_c} = 0$ with high probability. We denote ${{\bf F}'(\Ac(i))}^{(1,2,\ldots,\Ksf_{\rm c}(\Nsf_{\rm r}-1))_c}$ as ${\bf F}^{''}(\Ac(i))$. 

Let us suppose that it is possible to
make linear transformations on the rows of ${\bf S}^{\Ac(i)}$ into the form~\eqref{eq:S_ac}. 
\begin{figure*}
\begin{subequations}
\begin{align}
    {\bf S}^{\Ac(i)} 
    &= \left[ \begin{array}{cccc}
 c_{\Ac(i),1,1}  &  c_{\Ac(i),1,2}  & \cdots &  c_{\Ac(i),1,\Ksf_{\rm c}\Nsf_{\rm r}}   \\ 
 \vdots   & \vdots  &  \ddots& \vdots \\ 
 c_{\Ac(i),\Ksf_{\rm c},1}  &  c_{\Ac(i),\Ksf_{\rm c},2}  & \cdots &  c_{\Ac(i),\Ksf_{\rm c},\Ksf_{\rm c}\Nsf_{\rm r}}    
 \end{array}\right] \\
 & = \left[ \begin{array}{cccccccccc}
 c_{\Ac(i),1,1}  & \cdots &  c_{\Ac(i),1,\Ksf_{\rm c}(i-1)}  & 1 & 0 & \cdots & 0 & c_{\Ac(i),1,\Ksf_{\rm c}i +1 } & \cdots & c_{\Ac(i),1,\Ksf_{\rm c}\Nsf_{\rm r}}   \\ 
 c_{\Ac(i),2,1}  & \cdots &  c_{\Ac(i),2,\Ksf_{\rm c}(i-1)}  & 0 & 1 & \cdots & 0 & c_{\Ac(i),2,\Ksf_{\rm c}i +1 } & \cdots & c_{\Ac(i),2,\Ksf_{\rm c}\Nsf_{\rm r}}   \\ 
 \vdots  & \ddots & \vdots & \vdots & \vdots  &  \ddots& \vdots & \vdots  & \ddots & \vdots \\ 
 c_{\Ac(i),\Ksf_{\rm c},1}  & \cdots &  c_{\Ac(i),\Ksf_{\rm c},\Ksf_{\rm c}(i-1)}  & 0 & 0 & \cdots & 1 & c_{\Ac(i),\Ksf_{\rm c},\Ksf_{\rm c}i +1 } & \cdots & c_{\Ac(i),\Ksf_{\rm c},\Ksf_{\rm c}\Nsf_{\rm r}}      
 \end{array}\right]. \label{eq:S_ac}
\end{align}
\end{subequations}
\end{figure*}
Then from ${\bf S}^{\Ac(i)} {{\bf F}^{''}(\Ac(i))} = 0$, we will get 
\begin{subequations}
\begin{align}
   & {{\bf S}^{\Ac(i)}}^{([\Ksf_{\rm c}\Nsf_{\rm r}]\backslash[\Ksf_{\rm c}(i-1)+1: \Ksf_{\rm c}i])_c}  {{\bf F}^{''}(\Ac(i))}^{([\Ksf_{\rm c}\Nsf_{\rm r}]\backslash[\Ksf_{\rm c}(i-1)+1: \Ksf_{\rm c}i])_r} \\ &= - {{\bf F}^{''}(\Ac(i))}^{([\Ksf_{\rm c}(i-1)+1: \Ksf_{\rm c}i])_r} \\ &:= \left[ {\bf f}_{\Ac(i),\Ksf_{\rm c}(i-1)+1}; \cdots; {\bf f}_{\Ac(i),\Ksf_{\rm c}i} \right] . 
\end{align}
\end{subequations}
By the Cramer's rule, we can calculate:
$$ c_{\Ac(i),j,m} = \frac{\text{det}({\bf Y}_{\Ac(i),j,m})}{\text{det}({{\bf F}^{''}(\Ac(i))}^{([\Ksf_{\rm c}\Nsf_{\rm r}]\backslash[\Ksf_{\rm c}(i-1)+1: \Ksf_{\rm c}i])_r})},$$
where $j \in [\Ksf_{\rm c}(i-1)+1 : \Ksf_{\rm c}i]$, $m \in [\Ksf_{\rm c}\Nsf_{\rm r}]\backslash[\Ksf_{\rm c}(i-1)+1: \Ksf_{\rm c}i]$ and ${\bf Y}_{\Ac(i),j,m}$ is denoted as the matrix whose $j^{th}$ row is replaced by ${\bf f}_{\Ac(i),j}$. As a result, $c_{\Ac, j,m }$ can be presented by $\frac{P_{\Ac,j,m}}{Q_{\Ac,j,m}}$ where $P_{\Ac,j,m}$ and $Q_{\Ac,j,m}$ are multivariate polynomials whose variables come form the elements in ${\bf F}^{''}(\Ac(i))$.

Since ${\bf F}'$ is the left-null space of ${\bf E}$, and ${\bf F}^{''}(\Ac(i))$ is a sub-matrix with dimension $\Ksf_{\rm c}\Nsf_{\rm r} \times \Ksf_{\rm c}(\Nsf_{\rm r}-1)$, 
by the  Schwartz-Zippel Lemma~\cite{Schwartz,Zippel,Demillo_Lipton}, we have
\begin{subequations}
\begin{align}
    & \text{Pr}\{c_{\Ac(i),j,m} ~\text{exsits}\} \\
    & = \text{Pr}\{ \text{det}({{\bf F}^{''}(\Ac(i))}^{([\Ksf_{\rm c}\Nsf_{\rm r}]\backslash[\Ksf_{\rm c}(i-1)+1: \Ksf_{\rm c}i])_r})~\text{is not zero}\} \\
    & \ge 1 - \frac{\Ksf_{\rm c}(\Nsf_{\rm r} - 1)}{\qsf}.
\end{align}
\end{subequations}
By the probability union bound, we can further get:
\begin{subequations}
\begin{align}
    & \text{Pr}\{c_{\Ac(i),j,m} ~\text{exsits}, \nonumber \\ & \forall i \in [\Nsf_{\rm r}, j \in [\Ksf_{\rm c}, m \in [\Ksf_{\rm c}\Nsf_{\rm r}]\backslash [\Ksf_{\rm c}(i-1)+1:\Ksf_{\rm c}i]]]\} \\
    & \ge 1 - \frac{\Nsf\Ksf_{\rm c}\Ksf_{\rm c}(\Nsf_{\rm r} - 1)\Ksf_{\rm c}(\Nsf_{\rm r} - 1)}{\qsf} \\
    & = 1 - \frac{(\Ksf_{\rm c})^{3}\Nsf(\Nsf_{\rm r} - 1)}{\qsf} \\
    & \rightarrow 1 (\text{when}~\qsf \rightarrow \infty).
\end{align}
\end{subequations}

From the above discussion, the determent of the $\Ksf_{\rm c}\Nsf_{\rm r} \times \Ksf_{\rm c}\Nsf_{\rm r}$ matrix $ {\bf S}^{\Ac}$ can be denoted as 
\begin{align}
    \text{det}({\bf S}^{\Ac}) = \sum_{i \in [(\Ksf_{\rm c}\Nsf_{\rm r})!]}\frac{\Pc_i}{\Qc_i},
\end{align}
where each element in ${\bf S}^{\Ac}$ is expressed by $\frac{\Pc_i}{\Qc_i}$, a ratio of two polynomial whose variables come from the elemants in ${\bf F}'$.  In the Appendix~\ref{sec:SZ lemma proof}, we let $\text{det}({\bf S}^{\Ac}) = \frac{\Pc_{\Ac}}{\Qc_{\Ac}}$ and further discuss that if ${\bf S}^{\Ac}$ exits and $\Pc_{\Ac} \neq 0$, we will get the conclusion that $\text{det}({\bf S}^{\Ac}) \neq 0$ and thus ${\bf S}^{\Ac}$ is full-rank. Based on the Schwartz-Zippel Lemma~\cite{Schwartz,Zippel,Demillo_Lipton}, we want to prove $\Pc_{\Ac}$ is a non-zero multivariate polynomial. 

To apply the Schwartz-Zippel lemma~\cite{Schwartz,Zippel,Demillo_Lipton}, we need to guarantee that $\Pc_{i}$ is a non-zero multivariate polynomial. Since when $\Qc_{i} \neq 0$ there is $\Pc_{i} = \text{det}({{\bf S}^{\Ac}}){\Qc_{i}}$, we only need one specific realization of $\bf F$ so that $\text{det}({{\bf S}^{\Ac}}) \neq 0$ and $\Qc_{i} \neq 0$. 

\begin{lem}
    \label{lem: sz lemma situation}
    We have prove the decodability of the coding scheme in the following situations:
    \begin{itemize}
        \item In the two cases when $\Ksf_{\rm c} = \Nsf_{\rm r} - \msf +1$, and when $\Ksf = \Nsf = \Nsf_{\rm r}, \msf =2, \Nsf |(\Ksf_{\rm c} + 1),$ through Schwartz-Zippel lemma, we guarantee the decodability.
        \item When $\Nsf \le 60,$ we verify the decodability by numerically computation in practice.
    \end{itemize}

\end{lem}

\subsection{$\frac{\Ksf}{\Nsf}$ extension}
\label{sec: kN proof}
Using the Schwartz-Zippel lemma to prove the decodability, we need to find a specific realization of the elements in $\bf F$. When $\Nsf$ divides $\Ksf$, the system model can be established as similar to [Appendix A~\cite{Same_problem}]. More specific, assuming $a = \frac{\Ksf}{\Nsf}$, we construct demand matrix ${\bf F}$ as 
\begin{equation}
{\bf F} = \left[\begin{array}{c:c:c:c}
 ({\bf F}_1)_{\usf \times \Nsf}  & {\bf 0}_{\usf \times \Nsf}  & \cdots & {\bf 0}_{\usf \times \Nsf}   \\ \hdashline
{\bf 0}_{\usf \times \Nsf} &  ({\bf F}_2)_{\usf \times \Nsf}   & \cdots & {\bf 0}_{\usf \times \Nsf}   \\ \hdashline 
 \vdots   & \vdots  &  \ddots& \vdots \\ \hdashline
 {\bf 0}_{\usf \times \Nsf} &   {\bf 0}_{\usf \times \Nsf}    & \cdots &  ({\bf F}_{a})_{\usf \times \Nsf}  
 \end{array}
\right] .
\end{equation}
Under this construction, the $(\Ksf,\Nsf,\Nsf_{\rm r},\frac{\Ksf}{\Nsf}\usf,\msf)$ distributed linearly separable problem will be divided into $a$ independent $(\Nsf,\Nsf,\Nsf_{\rm r},\usf,\msf)$ sub-problems. If the proposed scheme works for $(\Nsf,\Nsf,\Ksf_{\rm c}, \msf)$ problem with high probability, the scheme will work for the $(\Ksf,\Nsf,\frac{\Ksf}{\Nsf}\usf, \msf)$ with high probability too.   

\section{Proofs for  the  Lemma~\ref{lem: sz lemma situation}}
\label{sec:SZ lemma proof}
From the Discussion in Appendix~\ref{sec: kN proof}, we only consider the $(  \Nsf,\Nsf,\Nsf_{\rm r},  \Ksf_{\rm c}, \msf )$ problem where $\Ksf_{\rm c}\in [ \Nsf_{\rm r}-\msf+1 ]$ in the following.

 \subsection{$ \Ksf_{\rm c}=\Nsf_{\rm r}-\msf+1$}
 \label{sub:SZ lemma proof Kc=Nr-m+1}
 We consider the case where $ \Ksf_{\rm c}=\Nsf_{\rm r}-\msf+1$. In this case, by dividing each message into $\dsf = \msf+\Ksf_{\rm c}-1$ sub-messages we have the effective demand as  
  $
   \mathbf{F}'{\bf W}
$
 where $\mathbf{F}' $ contains  $(\msf+\Ksf_{\rm c}-1)\Ksf_{\rm c}=\Nsf_{\rm r}\Ksf_{\rm c}$ rows; in other words, we do not need to add any virtual demand: 
 \begin{equation}
 \mathbf{F}'= \left[\begin{array}{c:c:c:c}
 ({\bf F})_{\Ksf_{\rm c} \times \Nsf}  & {\bf 0}_{\Ksf_{\rm c} \times \Nsf}  & \cdots & {\bf 0}_{\Ksf_{\rm c} \times \Nsf}   \\ \hdashline
{\bf 0}_{\Ksf_{\rm c} \times \Nsf} &  ({\bf F})_{\Ksf_{\rm c} \times \Nsf}   & \cdots & {\bf 0}_{\Ksf_{\rm c} \times \Nsf}   \\ \hdashline 
 \vdots   & \vdots  &  \ddots& \vdots \\ \hdashline
 {\bf 0}_{\Ksf_{\rm c} \times \Nsf} &   {\bf 0}_{\Ksf_{\rm c} \times \Nsf}    & \cdots &  ({\bf F})_{\Ksf_{\rm c} \times \Nsf}  
 \end{array}
\right].
 \end{equation}
 
 We then focus on a subset of workers $\Ac \subseteq [\Nsf]$ where $|\Ac|=\Nsf_{\rm r}$. 
 The received transmissions from the workers in $\Ac$ can be written as  $\mathbf{S_{\Ac}}  \mathbf{F}' [W_{1,1};\ldots;W_{\Nsf,1};W_{2,1};\ldots ; W_{\Nsf,\msf+\Ksf_{\rm c}-1}]$, where
 \begin{align}
 \mathbf{S_{\Ac}}= \left[\sv^{\Ac(1),1};\ldots;\sv^{\Ac(1),\Ksf_{\rm c}};\sv^{\Ac(2),1};\ldots;\sv^{\Ac(\Nsf_{\rm r}),\Ksf_{\rm c}} \right].
\label{eq:U vectors matrix11}
\end{align}

For each worker   $\Ac(i)$ where $i\in [\Nsf_{\rm r}]$, the number of messages it cannot compute is $\Nsf_{\rm r}-\msf=\Ksf_{\rm c}-1$.  Let $\overline{ {\bf F}_{\Ac(i)} }={\bf F}^{\left(\overline{  \Zc_{\Ac(i)} }  \right)_{\rm c}}$
be the column-wise sub-matrix of ${\bf F}$ containing the columns in $\overline{  \Zc_{\Ac(i)} } $.  The dimension of $\overline{ {\bf F}_{\Ac(i)} }$ is $\Ksf_{\rm c} \times (\Ksf_{\rm c}-1)$ and $\overline{ {\bf F}_{\Ac(i)} }$ is full-rank with high probability.
Hence, there exists one right null vector ${\bf s}_i$ with dimension $1\times \Ksf_{\rm c}$ of 
$\overline{ {\bf F}_{\Ac(i)} }$. Hence, $\sv^{\Ac(i),j}$ could be written as 
$$[a_{i,j,1}{\bf s}_i, a_{i,j,2}{\bf s}_i,\ldots, a_{i,j,\msf+\Ksf_{\rm c}-1}{\bf s}_i ],$$
where $a_{i,j,k}\in \mathbb{F}_{\qsf}$. 

Thus   $\mathbf{S_{\Ac}}$ could be divided into $\msf+\Ksf_{\rm c}-1$ blocks, 
 \begin{equation}\setstretch{1.25}
\mathbf{S_{\Ac}} =\begin{bmatrix}  \   
 \tikzmark{left1} \textcolor{white}{00}
a_{1,1,1}{\bf s}_1  & \tikzmark{left2} \textcolor{white}{0} a_{1,1,2}{\bf s}_1  & \cdots & \tikzmark{left3} \textcolor{white}{0}  a_{1,1,\msf+\Ksf_{\rm c}-1}{\bf s}_1    \\ \textcolor{white}{00} 
a_{1,2,1}{\bf s}_1   & \textcolor{white}{0} a_{1,2,2}{\bf s}_1   & \cdots & \textcolor{white}{0} a_{1,2,\msf+\Ksf_{\rm c}-1}{\bf s}_1   \\  \textcolor{white}{00}
 \vdots   & \vdots  &  \ddots& \vdots \\  \textcolor{white}{00}
a_{1,\Ksf_{\rm c},1}{\bf s}_1   & \textcolor{white}{0} a_{1,\Ksf_{\rm c},2}{\bf s}_1   & \cdots & \textcolor{white}{0} a_{1,\Ksf_{\rm c},\msf+\Ksf_{\rm c}-1}{\bf s}_1   \\  \textcolor{white}{00}
a_{2,1,1}{\bf s}_2   & \textcolor{white}{0} a_{2,1,2}{\bf s}_2   & \cdots & \textcolor{white}{0} a_{2,1,\msf+\Ksf_{\rm c}-1}{\bf s}_2   \\  \textcolor{white}{00}
 \vdots   & \vdots  &  \ddots& \vdots \\  \textcolor{white}{00}
 a_{\Nsf_{\rm r},\Ksf_{\rm c},1}{\bf s}_{\Nsf_{\rm r}}   \tikzmark{right1} & \textcolor{white}{0}  a_{\Nsf_{\rm r},\Ksf_{\rm c},2}{\bf s}_{\Nsf_{\rm r}} \tikzmark{right2}   & \cdots & \textcolor{white}{0}  a_{\Nsf_{\rm r},\Ksf_{\rm c},\msf+\Ksf_{\rm c}-1}{\bf s}_{\Nsf_{\rm r}}   \tikzmark{right3}
\end{bmatrix}.
\end{equation}
\DrawBox[thick, black,   dashed]{left1}{right1}{\textcolor{black}{\footnotesize${\bf F^{\prime}}_1$}}
\DrawBox[thick, red, dashed]{left2}{right2}{\textcolor{red}{\footnotesize${\bf F^{\prime}}_2$}}
\DrawBox[thick, blue, dashed]{left3}{right3}{\textcolor{blue}{\footnotesize${\bf F^{\prime}}_{\msf+\Ksf_{\rm c}-1}$}}

The determinant of $ \mathbf{S_{\Ac}}$ could be seen as $D_{\Ac}= \frac{P_{\Ac}}{Q_{\Ac}}$, where $P_{\Ac}$ and $Q_{\Ac}$ are multivariate polynomials whose variables are the elements in  $\mathbf{F}$ and $(a_{i,j,k} : i\in [\Nsf_{\rm r}], j\in [\Ksf_{\rm c}], k\in [\msf+\Ksf_{\rm c}-1])$.
 Since   each  element in   $\mathbf{F}$
  is uniformly and i.i.d. over $\mathbb{F}_{\qsf}$ where   $\qsf$ is large enough,
by the  Schwartz-Zippel Lemma~\cite{Schwartz,Zippel,Demillo_Lipton}, if we can further show that the multivariate polynomial  $P_{\Ac}$ is   non-zero      (i.e., a multivariate polynomial whose coefficients are not all $0$) implying $D_{\Ac}$ and $Q_{\Ac}$ are non-zero simultaneously, the probability that $P_{\Ac}$ is equal to $0$ over all possible realization of the elements in  $\mathbf{F}$ and $(a_{i,j,k} : i\in [\Nsf_{\rm r}], j\in [\Ksf_{\rm c}], k\in [\msf+\Ksf_{\rm c}-1])$ goes to $0$ when $\qsf$ goes to infinity, and thus 
  the matrix in~\eqref{eq:U vectors matrix11} is of full rank with high probability. 
 So in the following, we need to show that $P_{\Ac}$ is   non-zero, such that by receiving $\mathbf{S_{\Ac}}  \mathbf{F}' [W_{1,1};\ldots;W_{\Nsf,1};W_{2,1};\ldots ; W_{\Nsf,\msf+\Ksf_{\rm c}-1}]$ the master can recover $\mathbf{F}' [W_{1,1};\ldots;W_{\Nsf,1};W_{2,1};\ldots ; W_{\Nsf,\msf+\Ksf_{\rm c}-1}]$.

For each worker  $\Ac(i)$ and each $j\in [\Ksf_{\rm c}]$, we choose the coefficient  
$$a_{i,j, {\rm Mod}(i-1+j,\msf+\Ksf_{\rm c}-1) }=1,$$
while the other elements $a_{i,j, k}=0$ 
where $k\in [\msf+\Ksf_{\rm c}-1]\setminus \{{\rm Mod}(i-1+j,\msf+\Ksf_{\rm c}-1) \}$.
In other words, for each worker, $\Ac(i)$, the coefficients    in  $\msf-1$ blocks are all $0$ while in the remaining $\Ksf_{\rm c}$ blocks there is one row equal to $s_{i}$ and the remaining rows are ${\bf 0}_{1\times \Ksf_{\rm c}}$. 

Since in each row of $\mathbf{S_{\Ac}}$, there is exactly one block whose row is not ${\bf 0}_{1\times \Ksf_{\rm c}}$, to prove $\mathbf{S_{\Ac}}$ is full-rank we only need to prove that each block is full-rank in the column-wise. In other words, there are $\Ksf_{\rm c}$ rows in each block that are not ${\bf 0}_{1\times \Ksf_{\rm c}}$ and these rows are linearly independent. 
By construction, these  $ \Ksf_{\rm c} $ rows  are $s_{i}$ where $i\in \Bc_k$ and $|\Bc_k|=\Ksf_{\rm c}$. 
Next, we will prove that these  $\Ksf_{\rm c}$ rows are linearly independent with high probability.
The determinant of the matrix containing these  $\Ksf_{\rm c}$ rows is a multivariate polynomial 
 whose variables are the elements in  $\mathbf{F}$. 
 Following~\cite[Appendix D-A]{m=1}, we select the $i^{\text{th}}$ row of $\mathbf{F}$ as 
 \begin{align}
 [ *, *, \cdots,*, 0, 0, \cdots, 0, *,*,\cdots ,*], \label{eq:ith row} 
\end{align}
where
the $0$'s are located at the positions in $\overline{  \Zc_{\Ac(i)} }$  and each `$*$' represents a    symbol uniformly i.i.d. over $\mathbf{F}_{\qsf}$.
By a similar proof as in~\cite[Appendix D-A]{m=1}, this section of   $\mathbf{F}$ guarantees that 
the left null vectors $s_{i}$ where  $i\in \Bc_k$ exist and are linearly independent. 

Hence, we prove that  each block in $\mathbf{S_{\Ac}}$  is full-rank in the column-wise full-rank with high probability. By the probability union bound, we also have $\mathbf{S_{\Ac}}$ is full-rank with high probability.

\subsection{$\Ksf = \Nsf=\Nsf_{\rm r}$, and $\msf+ \Ksf_{\rm c}-1$ divides $\Nsf$}
\label{sub: the second case of SW lemma}

In this case, we also divide each message into $\dsf = \msf+\Ksf_{\rm c}-1$ sub-messages, and the effective demand is 
 $\setlength{\arraycolsep}{1.2pt}
\mathbf{F}' {\bf W} 
$.
And we will construct ${\bf F}'$ like this
\begin{equation}
{\bf F}'=
\left[\begin{array}{c:c:c:c}
 ({\bf F})_{\Ksf_{\rm c} \times \Nsf}  & {\bf 0}_{\Ksf_{\rm c} \times \Nsf}  & \cdots & {\bf 0}_{\Ksf_{\rm c} \times \Nsf}   \\ \hdashline
{\bf 0}_{\Ksf_{\rm c} \times \Nsf} &  ({\bf F})_{\Ksf_{\rm c} \times \Nsf}   & \cdots & {\bf 0}_{\Ksf_{\rm c} \times \Nsf}   \\ \hdashline 
 \vdots   & \vdots  &  \ddots& \vdots \\ \hdashline
 {\bf 0}_{\Ksf_{\rm c} \times \Nsf} &   {\bf 0}_{\Ksf_{\rm c} \times \Nsf}    & \cdots &  ({\bf F})_{\Ksf_{\rm c} \times \Nsf} \\ \hdashline
 ({\bf V_ {1}})_{\Ksf_{\rm c}\Usf \times \Nsf}  &  ({\bf V_ {2}})_{\Ksf_{\rm c}\Usf \times \Nsf}  &   \cdots   &   ({\bf V_ {\dsf }})_{\Ksf_{\rm c}\Usf \times \Nsf} 
 \end{array}
\right]_{\Ksf_{\rm c}\Nsf \times \Nsf\dsf} ,
\end{equation}
 where $\Usf = \Nsf_{\rm r} - \msf - \Ksf_{\rm c} + 1 = \Nsf - \msf - \Ksf_{\rm c} + 1$ and $\dsf=\msf+\Ksf_{\rm c}-1$.
 We construct $({\bf V_ i})_{\Ksf_{\rm c}\Usf \times \Nsf}, i \in [\Ksf_{\rm c} + \msf - 1]$ from the null space of the rank reduction equation $\bf E$ similarly to Section IV. Next, we will show how to construct $\bf F$.
 
\hspace*{-4mm} $\bf E$ is a $\Nsf(\msf - 1) \times \Nsf\dsf$ matrix. Since $ \Nsf = \Nsf_{\rm r} ,\Zsf = \frac{\Nsf}{\msf+\Ksf_{\rm c}-1}, \Gsf = \frac{\Nsf(\msf-1)}{\msf+\Ksf_{\rm c}-1}, $ the dimension of $\bf E$ can be written as $\Gsf(\msf+\Ksf_{\rm c}-1) \times \Nsf(\msf+\Ksf_{\rm c}-1)$. We design the matrix $\bf E$ as :
\begin{align*}
\left[\begin{array}{c:c:c:c}
 ({\bf E}_{1})_{\Gsf \times \Nsf}  & {\bf 0}_{\Gsf \times \Nsf}  & \cdots & {\bf 0}_{\Gsf \times \Nsf}   \\ \hdashline
{\bf 0}_{\Gsf \times \Nsf} &  ({\bf E}_2)_{\Gsf \times \Nsf}   & \cdots & {\bf 0}_{\Gsf \times \Nsf}   \\ \hdashline 
 \vdots   & \vdots  &  \ddots& \vdots \\ \hdashline
 {\bf 0}_{\Gsf \times \Nsf} &   {\bf 0}_{\Gsf \times \Nsf}    & \cdots &  ({\bf E}_{\msf+\Ksf_{\rm c}-1})_{\Gsf \times \Nsf}  
 \end{array}
\right],
\end{align*}
with dimension ${\Gsf(\msf+\Ksf_{\rm c}-1) \times \Nsf(\msf+\Ksf_{\rm c}-1)}$, and ${\bf E}_{i}, \forall i \in  [\msf+\Ksf_{\rm c} - 1]$ are in the diagonal.

\hspace*{-4mm} We further construct  ${\bf E}_{i}, \forall i \!\in \![\msf+\Ksf_{\rm c} - 1]$ in the following: Since $\Nsf$ can divide $\msf+\Ksf_{\rm c} - 1$, we can partition $\bf F$ into $\Zsf \!\!=\!\! \frac{\Nsf}{\msf+\Ksf_{\rm c}-1}$ blocks and then each block is a $\Ksf_{\rm c} \!\times \!(\msf+\Ksf_{\rm c} - 1)$ sub-matrix.\ More specifically, each pair of $\msf+\Ksf_{\rm c} - 1$ neighbouring columns in $\bf F$ constitute a sub-matrix which is denoted as $\{{\bf F}^{ (1+(j-1)(\msf+\Ksf_{\rm c} -1))_c}, {\bf F}^{ (2+(j-1)(\msf+\Ksf_{\rm c} -1))_c},  \cdots, $ $ {\bf F}^{ (\msf+\Ksf_{\rm c}-1+(j-1)(\msf+\Ksf_{\rm c} -1))_c} \} = {\bf F}^{(j)_c}$, where $j$ represents the number of each block and $j \in [\Zsf]$. 

\hspace*{-3mm} Furthermore, we use the character of Block $j$ to design $({\bf E}_i)_{\Gsf \times \Nsf}$. For example, sub-matrix ${\bf F}^{(1)_c}$ has $\msf-1$ non-zero linearly independent vectors $\mathbf{y}_k$ where $k \in [\msf-1]$ so that ${\bf F}^{(1)_c}{\mathbf{y}}_{k}^{\text{\rm T}} = {\bf 0}_{\Ksf_{\rm c} \times 1}$. And $\mathbf{y}_k=(0,0,\dots,0,*,*,\dots,*,0,0,\dots,0)$, where the first non-zero entry `$*$' is located at the $k-1$ column of $\mathbf{y}_k$.

\hspace*{-4mm} We construct a sub-matrix ${\bf E}_{1,1}$,which can be constructed as:
\begin{align*}
\mathbf{E_{1,1}} 
=\left[\begin{array}{c}
(a_{1}y_{1})_{1 \times (\msf+\Ksf_{\rm c}-1)} \\ \hdashline
(a_{2}y_{2})_{1 \times (\msf+\Ksf_{\rm c}-1)}   \\ \hdashline 
\vdots \\ \hdashline
(a_{\msf-1}y_{\msf-1})_{1 \times (\msf+\Ksf_{\rm c}-1)}  
\end{array}
\right]_{(\msf-1) \times (\msf+\Ksf_{\rm c}-1)},
\end{align*}
where $a_{k}$ is randomly selected over the field $\bf F_{\rm q}$.Obviously,the sub-matrix ${\bf E}_{1,1}$ is full-rank due to the linear independence between $\mathbf{y}_k$. Similarly, the sub-matrix ${\bf E}_{i,j}$ is full-rank. With the same progress, we constructed the matrix $\bf E$.

\hspace*{-4mm} More specifically, ${\bf E}_{i}$ can always be written as 
\begin{equation}
\left[\begin{array}{c:c:c:c}
 ({\bf E}_{i,1})_{(\msf-1) \times \dsf}  & {\bf 0}_{(\msf-1) \times \dsf}  & \cdots & {\bf 0}_{(\msf-1) \times \dsf}   \\ \hdashline
{\bf 0}_{(\msf-1) \times \dsf} &  ({\bf E}_{i,2})_{(\msf-1) \times \dsf}   & \cdots & {\bf 0}_{(\msf-1) \times \dsf}   \\ \hdashline 
 \vdots   & \vdots  &  \ddots& \vdots \\ \hdashline
 {\bf 0}_{(\msf-1) \times \dsf} &   {\bf 0}_{(\msf-1) \times \dsf}    & \cdots &  ({\bf E}_{i,\Zsf})_{(\msf-1) \times \dsf}  
 \end{array}
\right].
\end{equation} 

\hspace*{-4mm} Note that rank reduction matrix $\bf E$ here is a special case in the section IV.
Because the non-zero elements in each row are totally in different locations, $\bf E$ is full-rank.
As a result, the dimension of the null space of matrix $\bf E$ is  $ \Ksf_{\rm c}\Nsf \times \Nsf(\msf+\Ksf_{\rm c} - 1)  $ which is equal to the dimension of matrix ${\bf F}'$. 

\hspace*{-4mm} Let's denote $\bf S {\bf F}'$ as matrix $\bf Q$, then $\bf Q$ is a $ \Ksf_{\rm c}\Nsf \times \Nsf(\msf+\Ksf_{\rm c} - 1)$ matrix and in the null space of matrix $\bf E$. If it is possible to make a full-rank $\bf Q$ matrix, the matrix $\bf S$ is full-rank since the matrix $\bf S$ is a transfer matrix from ${\bf F}'$ to $\bf Q$ and ${\bf F}'$, $\bf Q$ are both full-rank matrices. There is a constraint that ${\bf Q} {\bf E}^{\text{\rm T}} = {\bf 0}$.

\hspace*{-4mm} We construct $\bf Q$ like:
\begin{equation}
\left[\begin{array}{c:c:c:c}
 ({\bf Q}_{1})_{(\Nsf-\Gsf)\times \Nsf}  & {\bf 0}_{(\Nsf-\Gsf) \times \Nsf}  & \cdots & {\bf 0}_{(\Nsf-\Gsf) \times \Nsf}   \\ \hdashline
{\bf 0}_{(\Nsf-\Gsf) \times \Nsf} &  ({\bf Q}_2)_{(\Nsf-\Gsf) \times \Nsf}   & \cdots & {\bf 0}_{(\Nsf-\Gsf) \times \Nsf}   \\ \hdashline 
 \vdots   & \vdots  &  \ddots& \vdots \\ \hdashline
 {\bf 0}_{(\Nsf-\Gsf) \times \Nsf} &   {\bf 0}_{(\Nsf-\Gsf) \times \Nsf}    & \cdots &  ({\bf Q}_{\msf+\Ksf_{\rm c}-1})_{(\Nsf-\Gsf) \times \Nsf}  
 \end{array}
\right] ,
\end{equation}
where $({\bf Q}_{i})_{(\Nsf-\Gsf) \times \Nsf}, i \in [\msf+\Ksf_{\rm c}-1]$ is computed from the null space of $({\bf E}_{i})_{\Gsf \times \Nsf}, i \in [\msf+\Ksf_{\rm c}-1]$ and since there is no overlap of non-zero locations among rows in $({\bf E}_{i})_{\Gsf \times \Nsf}$, we can easily construct full-rank $({\bf Q}_{i})_{(\Nsf-\Gsf) \times \Nsf}$ where each row has $\Nsf_{\rm r}-\msf=\Nsf-\msf$ zero elements. 

\hspace*{-3mm} For instance, ${\bf Q}_{i,1}$ is corresponding to the null space of$({\bf E}_{i,1})_{(\msf-1) \times (\msf+\Ksf_{\rm c}-1)}$. It can be written as 
\begin{equation}\setstretch{1.5}
{\bf Q}_{i,1} = \begin{bmatrix} 
 \tikzmark{left1} \textcolor{white}{00}
  * & \cdots & * & 0 &\cdots &0 & \tikzmark{left2} \textcolor{white}{0} 0 & \cdots & 0   \\ 
  \textcolor{white}{00} 
  0  & * &\cdots & * & 0 & 0 & \textcolor{white}{0} 0 &\cdots &0   \\
  \textcolor{white}{00} 
   0  & 0 & *  & \cdots & * & 0 & \textcolor{white}{0} 0 &\cdots &0   \\
  
  \textcolor{white}{00}
  \vdots  & \vdots  &\vdots & \vdots & \vdots & \vdots& \vdots & \ddots & \vdots  \\ 
  \textcolor{white}{00}
   0 & \cdots & 0  & * & \cdots & * \tikzmark{right1} & \textcolor{white}{0} 0 &\cdots &0 \tikzmark{right2}  \\
\end{bmatrix}_{\Ksf_{\rm c} \times \Nsf} ,
\end{equation}
\DrawBox[thick, black, dashed]{left1}{right1}{\textcolor{black}{\footnotesize$\msf+\Ksf_{\rm c}-1$}}
\DrawBox[thick, red, dashed]{left2}{right2}{\textcolor{red}{\footnotesize$(\Zsf-1)(\msf+\Ksf_{\rm c}-1)$}}

\hspace*{-6mm} in which there are total $\msf$ non-zero entries in each row. We find that all the non-zero entries are in the first $(\msf+\Ksf_{\rm c}-1)$ columns.Similarly, each ${\bf Q}_{i,j}, j \in [\Zsf]$ has non-zero entries in the columns from $(1+(\msf+\Ksf_{\rm c}+1)(j-1))$ to $(\msf+\Ksf_{\rm c}-1)j$. Since ${\bf Q}_{i,j}, j \in [\Zsf]$ is generated from zero space, it is naturally of full rank.And $({\bf Q}_{i})_{(\Nsf-\Gsf)\times \Nsf}$ can be written as 
\begin{equation}
{\bf Q}_{i} = \begin{bmatrix}
({\bf Q}_{i,1})_{\Ksf_{\rm c} \times \Nsf}\\
({\bf Q}_{i,2})_{\Ksf_{\rm c} \times \Nsf}\\
\vdots\\
({\bf Q}_{i,\Zsf})_{\Ksf_{\rm c} \times \Nsf}
\end{bmatrix}.
\end{equation}
Due to different locations of non-zero entries among ${\bf Q}_{i,j}, j \in [\Zsf]$, ${\bf Q}_{i}$ is a full-rank matrix too.  

\hspace*{-3mm} With the similar observation above, according to the location of non-zero entries,  rows in ${\bf Q}_{i}, i \in [\msf+\Ksf_{\rm c}-1]$ have non-zero entries in different and non-overlapping columns, and ${\bf Q}_{i}, i \in [\msf+\Ksf_{\rm c}-1]$ is full-rank. As a result, $\bf Q$ is full-rank so that $\bf S$ is full-rank too.

\section{Proof for Theorem~\ref{cor: rep assignment}}
\label{sec:repetition proof}
Under the repetition assignment, the number of workers $\Nsf$ should be a multiple of $\Nsf - \Nsf_{\rm r} + \msf$. We divide the $\Nsf$ workers into $\frac{\Nsf}{\Nsf - \Nsf_{\rm r} + \msf}$ groups, each consisting of $\Nsf - \Nsf_{\rm r} + \msf$ workers. 
Within each group,  the same $\frac{\Ksf}{\Nsf}(\Nsf - \Nsf_{\rm r} + \msf)$ datasets are assigned to each worker, where there does not exist any common  dataset assigned to more than one group. 
Thus the datasets assigned to each worker  in group $i$ is 
\begin{align}
&\Zc'_i:= \underset{p \in \left[0:  \frac{\Ksf}{\Nsf} -1 \right]}{\bigcup}   \big\{\text{Mod}(1+(\Nsf - \Nsf_{\rm r} + \msf)(i-1),\Nsf)+ p  \Nsf , \\& \text{Mod}(2+(\Nsf - \Nsf_{\rm r} + \msf)(i-1),\Nsf)+\nonumber  p  \Nsf , \ldots, \\& \text{Mod}((\Nsf - \Nsf_{\rm r} + \msf)+(\Nsf - \Nsf_{\rm r} + \msf)(i-1),\Nsf)+ p  \Nsf  \big\}.\label{eq:Zni}
\end{align}
We denote the $n^{\text{th}}$ worker in the  group  $i$ as worker $(n, i)$ where $n\in [\Nsf - \Nsf_{\rm r} + \msf]$ and $i\in \left[\frac{\Nsf}{\Nsf - \Nsf_{\rm r} + \msf}\right]$; the set of datasets assigned to worker $(n,i)$ is $\Zc_{n,i}=\Zc'_i$  in~\eqref{eq:Zni}.


\paragraph{Converse bound for Theorem~\ref{cor: rep assignment}}
Since the system can tolerate any $\Nsf - \Nsf_{\rm r}$ stragglers among the $\Nsf$ workers, when all the stragglers are in group $i$, given  the codewords from any $\Nsf - \Nsf_{\rm r} + \msf - (\Nsf - \Nsf_{\rm r}) = \msf$ workers in group $i$ (denoted by $\Ac_i \in \binom{[\Nsf - \Nsf_{\rm r} + \msf]}{\msf} $) and the messages in $[\Ksf]\setminus \Zc'_i$, we can recover  ${\bf F}^{(\Zc_{n,i})_{c}}[W_{V};\ldots;W_{V+\Ksf-\Nsf}]$ (where $V = \text{Mod}(1+(\Nsf - \Nsf_{\rm r} + \msf)(i-1),\Nsf)$), which is with rank $\min\{\Ksf_{\rm c}, \frac{\Ksf}{\Nsf}(\Nsf - \Nsf_{\rm r} +\msf)\}$ with high probability; thus  
\begin{align}
    \sum_{n \in \Ac_i}\Tsf_{n,i} &\ge H({\bf F}[W_1;W_2;\ldots;W_{\Ksf}] | [\Ksf]\setminus \Zc'_i ) \\ &= \min\{\Ksf_{\rm c}, \frac{\Ksf}{\Nsf}(\Nsf - \Nsf_{\rm r} +\msf)\} \Lsf,  \label{eq:rep converse ineq}
\end{align}
where $\Ac_{i}  \in \binom{[\Nsf - \Nsf_{\rm r} + \msf]}{\msf}$ and $\Tsf_{n,i}$ represents the length of codewords from worker $(n,i)$. 
Considering all $\binom{\Nsf - \Nsf_{\rm r} + \msf}{\msf}$ different  sets $\Ac_i$ in $\binom{[\Nsf - \Nsf_{\rm r} + \msf]}{\msf}$ and   summing these inequalities as in~\eqref{eq:rep converse ineq}, we have 
\begin{align}
    \sum_{n \in [\Nsf - \Nsf_{\rm r} + \msf]}\Tsf_{n,i} \ge \frac{\Nsf - \Nsf_{\rm r} + \msf}{\msf}\min\{\Ksf_{\rm c}, \frac{\Ksf}{\Nsf}(\Nsf - \Nsf_{\rm r} +\msf)\}\Lsf.  
    \label{eq:one group}
\end{align}

Considering all groups and summing the inequalities as in~\eqref{eq:one group}, we have 
\begin{subequations}
\begin{align}
 & \sum_{i\in \left[\frac{\Nsf}{\Nsf - \Nsf_{\rm r} + \msf}\right]}  \sum_{n \in [\Nsf - \Nsf_{\rm r} + \msf]}\Tsf_{n,i} \ge \frac{\Nsf}{\msf}\min\{\Ksf_{\rm c}, \frac{\Ksf}{\Nsf}(\Nsf - \Nsf_{\rm r} +\msf)\}\Lsf;  
    \label{eq:all group}\\
 &   \Longrightarrow   \Rsf^{\star}_{\text{rep}} = \mathop {\max }\limits_{\Ac \subseteq [\Nsf]:|\Ac| = {\Nsf_{\rm r}}} \frac{{\sum\nolimits_{n \in \Ac} {{\Tsf_n}} }}{\Lsf} \nonumber \\ & \quad \quad \quad \quad \geq  \frac{\Nsf_{\rm r}}{\msf}\min\{\Ksf_{\rm c}, \frac{\Ksf}{\Nsf}(\Nsf - \Nsf_{\rm r} +\msf)\}.
 \label{eq:converse bound under rep}
\end{align}
\end{subequations}

\paragraph{Achievable coding scheme for Theorem~\ref{cor: rep assignment}}
We then present a  straightforward scheme that achieves the minimum communication cost~\eqref{eq:converse bound under rep}.

 If $\Ksf_{\rm c}\leq \frac{\Ksf}{\Nsf}(\Nsf - \Nsf_{\rm r} +\msf)$, we let each worker in group $i$ transmit $\frac{\Ksf_{\rm c}}{\msf}\Lsf$ random linear combinations of the $\Ksf_{\rm c}\Lsf$ symbols in ${\bf F}^{(\Zc_{n,i})_{c}}[W_{V};\ldots;W_{V+\Ksf-\Nsf}]$. Since $\Lsf$ is large enough, by receiving the transmissions of any $\msf $ workers in group $i$, we can recover ${\bf F}^{(\Zc_{n,i})_{c}}[W_{V};\ldots;W_{V+\Ksf-\Nsf}]$. By summing the above $\Ksf_{\rm c}\Lsf$ symbols from all groups, we can recover the computation task.

If $\Ksf_{\rm c}> \frac{\Ksf}{\Nsf}(\Nsf - \Nsf_{\rm r} +\msf)$, we let each worker in group $i$ transmit  $\frac{\Ksf (\Nsf - \Nsf_{\rm r} +\msf)}{\Nsf \msf} \Lsf $ random linear combinations of $\frac{\Ksf}{\Nsf}(\Nsf - \Nsf_{\rm r} +\msf) \Lsf$ symbols in $(W_{V},\ldots,W_{V+\Ksf-\Nsf})$.  Since $\Lsf$ is large enough, by receiving the transmissions of any $\msf $ workers in group $i$, we can recover $(W_{V},\ldots,W_{V+\Ksf-\Nsf})$. Thus by receiving the transmissions from all groups, we can recover all the messages and thus recover the computation task.

\bibliographystyle{IEEEtran}
\bibliography{re} 
\clearpage
\end{sloppypar}
\end{document}

%% file: macros.tex
\long\def\comment#1{}

\newcommand{\ev}{{\mathbf e}}

\newcommand{\sv}{{\mathbf s}}

\newcommand{\Ac}{{\mathcal A}}
\newcommand{\Bc}{{\mathcal B}}
\newcommand{\Cc}{{\mathcal C}}
\newcommand{\Dc}{{\mathcal D}}

\newcommand{\Gc}{{\mathcal G}}
\newcommand{\Hc}{{\mathcal H}}

\newcommand{\Kc}{{\mathcal K}}

\newcommand{\Nc}{{\mathcal N}}

\newcommand{\Pc}{{\mathcal P}}
\newcommand{\Qc}{{\mathcal Q}}
\newcommand{\Rc}{{\mathcal R}}
\newcommand{\Sc}{{\mathcal S}}

\newcommand{\Zc}{{\mathcal Z}}

\newcommand{\dsf}{{\mathsf d}}

\newcommand{\fsf}{{\mathsf f}}
\newcommand{\gsf}{{\mathsf g}}
\newcommand{\hsf}{{\mathsf h}}
\newcommand{\isf}{{\mathsf i}}

\newcommand{\lsf}{{\mathsf l}}
\newcommand{\msf}{{\mathsf m}}

\newcommand{\qsf}{{\mathsf q}}

\newcommand{\ssf}{{\mathsf s}}

\newcommand{\usf}{{\mathsf u}}

\newcommand{\vsf}{{\mathsf v}}

\newcommand{\Gsf}{{\mathsf G}}
\newcommand{\Hsf}{{\mathsf H}}

\newcommand{\Ksf}{{\mathsf K}}
\newcommand{\Lsf}{{\mathsf L}}
\newcommand{\Msf}{{\mathsf M}}
\newcommand{\Nsf}{{\mathsf N}}

\newcommand{\Rsf}{{\mathsf R}}
\newcommand{\Ssf}{{\mathsf S}}
\newcommand{\Tsf}{{\mathsf T}}
\newcommand{\Usf}{{\mathsf U}}

\newcommand{\Xsf}{{\mathsf X}}

\newcommand{\Zsf}{{\mathsf Z}}

\def\fsf{{\mathsf f}}

\newtheorem{thm}{Theorem}
\newtheorem{cor}{Corollary}
\newtheorem{lem}{Lemma}

\newtheorem{rem}{Remark}
\newtheorem{example}{Example}

\providecommand{\definitionname}{Definition}